\documentclass[sigconf,nonacm]{acmart}

\usepackage{amsmath}
\usepackage{algorithmic}

\AtEndPreamble{	\theoremstyle{acmplain}
	
	\theoremstyle{acmdefinition}

}

\usepackage[linesnumbered,vlined,ruled]{algorithm2e}
\DontPrintSemicolon
\SetKwComment{note}{$\triangleright$ }{}
\SetFuncSty{textsc}
\SetCommentSty{textit}
\SetDataSty{texttt}
\SetKwInput{Initial}{Initial}
\SetKwInput{Parameter}{Param.}
\SetKwInput{Input}{Input}
\SetKwInput{Data}{Data}
\SetKwInput{Output}{Output}
\ResetInOut{Result}
\let\oldnl\nl
\newcommand{\nonl}{\renewcommand{\nl}{\let\nl\oldnl}}
\SetKw{KwAnd}{and}
\SetKw{KwOr}{or}
\SetKw{KwXor}{xor}
\SetKw{KwNot}{not}
\SetKw{Parallel}{parallel}
\SetKw{Return}{return}
\SetKw{Break}{break}
\SetKw{Continue}{continue}

\usepackage{xspace}
\usepackage{booktabs}
\usepackage{multirow}
\usepackage{thmtools}
\usepackage{cleveref}
\usepackage{subcaption}

\AtBeginDocument{	}

\DeclareMathOperator*{\argmin}{arg\,min}

\DeclareMathOperator{\Unfair}{Unfair}

\newcommand{\red}{\textcolor{red}{red}\xspace}
\newcommand{\blue}{\textcolor{blue}{blue}\xspace}

\newcommand{\OPT}{\text{OPT}\xspace}
\DeclareRobustCommand{\alg}{\textsc{Dense}\xspace}
\DeclareRobustCommand{\algfast}{\textsc{Blade}\xspace}

\DeclareRobustCommand{\blopt}{\textsc{Opt}\xspace}

\DeclareRobustCommand{\blgapgreedy}{\textsc{GapGreedy}\xspace}
\DeclareRobustCommand{\algslow}{\textsc{Blade-NoBatch}\xspace}
\DeclareRobustCommand{\blkatzmass}{\textsc{KatzMass}\xspace}

\DeclareRobustCommand{\blsamegroup}{\textsc{SameGroup}\xspace}
\DeclareRobustCommand{\blfairwalk}{\textsc{FairWalk}\xspace}
\DeclareRobustCommand{\blcrosswalk}{\textsc{CrossWalk}\xspace}
\DeclareRobustCommand{\blfairgd}{\textsc{FairGD}\xspace}
\DeclareRobustCommand{\bllfprn}{\textsc{LFPR-N}\xspace}
\DeclareRobustCommand{\bllfpru}{\textsc{LFPR-U}\xspace}

\crefname{assumption}{assumption}{assumptions}
\Crefname{assumption}{Assumption}{Assumptions}

\begin{document}

\title{Fair Top-k Katz Centrality via Graph Design}
\author{Ivan Qin}
\email{ivanqin@liverpool.ac.uk}
\affiliation{	\institution{University of Liverpool}
	\city{Liverpool}
	\country{UK}
}
\author{Prudence Wong}
\email{pwong@liverpool.ac.uk}
\affiliation{	\institution{University of Liverpool}
	\city{Liverpool}
	\country{UK}
}
\author{Lutz Oettershagen}
\email{lutz.oettershagen@liverpool.ac.uk}
\affiliation{	\institution{University of Liverpool}
	\city{Liverpool}
	\country{UK}
}

\begin{abstract}
	Centrality measures are widely used to rank nodes in networked data, but
	fairness interventions for graph centrality typically target global score mass
	or modify the centrality operator rather than controlling who appears in the
	displayed top-$k$ ranking. We study this top-$k$ setting for Katz centrality.
	Given a target group proportion, an admissible set of directed edge additions,
	and a fairness tolerance, the goal is to find the smallest edit set whose
	resulting Katz top-$k$ ranking satisfies the target representation constraint.
	We formalize this problem as Fair Top-$k$ Katz Centrality Design and show that
	the minimum-edit objective is strongly inapproximable, ruling out worst-case
	polynomial-time approximation guarantees unless $\mathrm{P}=\mathrm{NP}$. We then derive
	closed-form Katz sensitivity expressions showing that useful edits are
	boundary-driven: they must help promotable nodes outside the top-$k$ set
	overtake opposing nodes inside it. Based on this structure, we develop \algfast,
	a scalable boundary-link algorithm that avoids dense Katz-kernel maintenance by
	using score-based direct-target batches and warm-started Katz updates.
	Experiments on synthetic and real-world networks show that \algfast reaches the desired top-$k$ representation using far fewer edits than natural baselines, while scaling to large real-world graphs and preserving the original ranking structure.
\end{abstract}



\maketitle

\section{Introduction}
\label{sec:introduction}

Centrality measures are fundamental tools for ranking nodes in networked data. They are used to identify influential users, important items, authoritative sources, or promising candidates for downstream decisions such as recommendation, outreach, moderation, and resource allocation~\cite{newman2018networks,borgatti2005centrality,freeman1978centrality,brin1998anatomy,gleich2015pagerank}. Katz centrality is a particularly prominent walk-based measure: it assigns importance to a node by aggregating the contributions of all walks ending at that node, while exponentially discounting longer walks~\cite{katz1953new,bonacich1987power,nathan2017dynamic}.

In many applications, however, centrality scores are not consumed as a
complete vector over all nodes. Instead, they are used to produce a short
list: the top users to contact, the top items to recommend, the top accounts
to audit, or the top candidates to display~\cite{oettershagen2022computing,
zhan2017identification,shi2019realtime,niu2012top}. Prior work has shown that
group size, homophily, group mixing, and systematic errors in observed edges
can substantially affect minority representation in centrality-based
rankings, including representation among the highest-ranked
nodes~\cite{karimi2018homophily,neuhauser2021simulating,
espin2022inequality,oliveira2022group,shen2025minority}.
This makes the composition of the displayed top-$k$ set a fairness object in
its own right. A group may receive substantial centrality mass across many
lower-ranked nodes while still being almost absent from the top-$k$ set.
We therefore study the realized group composition of the Katz top-$k$ set.

Crucially, top-$k$ representation is a boundary phenomenon: it changes only when nodes cross the selection cutoff. Improving aggregate group centrality is therefore in general not sufficient to guarantee a change in the displayed set. Our objective is to make the protected-group share among the $k$ highest-scoring nodes reach a prescribed target composition.

\textbf{Intervention model.}
We assume that the ranking mechanism is fixed and cannot be modified, while the underlying network can be influenced through admissible directed edge additions. Thus, we keep the Katz centrality rule fixed and intervene only on the graph. The admissible edge set is application-defined and specifies which interventions are allowed, while the number of added edges measures intervention cost. The central question is:
\emph{which minimum-size set of admissible graph edits makes the Katz top-$k$ set reach a prescribed target composition?}

\begin{figure}[t]
	\centering
	\includegraphics[width=0.9\columnwidth]{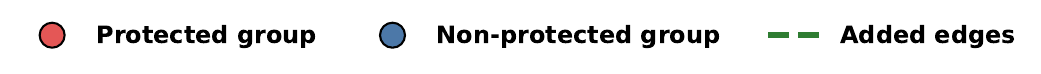}
	\includegraphics[width=\linewidth]{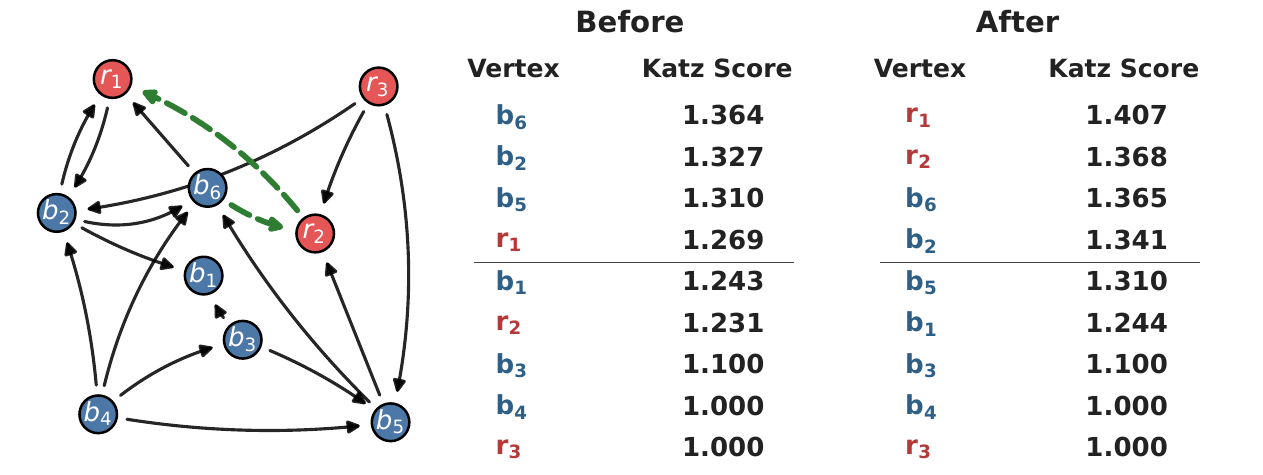}
	\caption{A toy example of the Fair Top-4 Katz Centrality Design Problem. Two edge additions (green dashed) are sufficient to reach a fair (balanced) top-4 group composition.}
	\label{fig:intro_example}
\end{figure}

Graph modifications have previously been studied as a means to improve the
centrality or rank of designated nodes or groups, as well as to improve
group-level information-flow fairness~\cite{bergamini2018improving,medya2018group,jalali2020information}.
Our objective differs in asking for the minimum number of admissible edits
needed to satisfy a prescribed group-composition constraint in the top-$k$ set induced by the fixed Katz centrality rule.

This problem arises naturally in networked ranking systems where
the ranking mechanism is fixed, but some connections in the underlying graph
can be influenced to improve group representation in the top-$k$. Examples
include creator discovery, marketplace seller ranking, scholarly search,
expert search, and professional networking. Concretely, a creator platform may
promote a top-$k$ set of accounts using follower or interaction-network signals,
while collaborations, cross-promotions, or recommended connections can create
new edges that help underrepresented creators enter the displayed set.
Similarly, a marketplace may rank sellers using transaction or trust-network
information while facilitating new buyer--seller interactions or partnerships
to improve the representation of a target seller group, and a scholarly search
system may rank papers or authors using citation-network features while
institutions can foster collaborations or dissemination links that improve the
top-$k$ representation of a target group of researchers.

\Cref{fig:intro_example} shows an example. Consider the directed network before adding the new green dashed edges.
The nodes are partitioned into a protected group, shown in red, and a
non-protected group, shown in blue. In the original graph, the top-4 nodes
ranked by Katz centrality are $b_6$, $b_2$, $b_5$, and $r_1$, so only
one protected node appears in the top-4 set. 
After adding the two new edges (green, dashed), the top-4 set becomes
$r_1$, $r_2$, $b_6$, and $b_2$. Thus, the protected count in the
top-4 increases from one to two, and the fair target composition is reached.
The example illustrates the boundary nature of the task: the intervention does
not need to redistribute Katz mass globally, but only to move the right nodes
across the top-$k$ cutoff.\looseness=-1

\textbf{Our approach.}
We formalize this task as the \emph{Fair Top-$k$ Katz Centrality Design
	Problem}. Given a directed graph, a set of admissible edge additions, a
top-$k$ size, a target group proportion, and a fairness tolerance, the
objective is to find the minimum-cardinality edit set whose induced Katz top-$k$
ranking satisfies the target tolerance.

Our technical approach exploits the algebraic structure of Katz centrality.
Since Katz centrality admits both a resolvent representation and a
walk-summation interpretation~\cite{katz1953new,bonacich1987power,benzi2015limiting,estrada2008communicability,nathan2017dynamic},
we can analyze how a single edge addition changes centrality scores and
ranking gaps. These sensitivity expressions allow us to reason directly about
the top-$k$ boundary: which promoted nodes outside the top-$k$ can overtake
which opposing nodes inside the top-$k$, and which admissible edits reduce the
corresponding score gaps.

Building on this analysis, we first develop a dense frontier-search reference algorithm which
implements the boundary principle directly: it evaluates exact pairwise Katz
differential gains and commits only edit prefixes that improve the top-$k$
objective or safely reduce a boundary gap. The dense algorithm is useful
as a principled reference method, but it also creates the main scalability
bottleneck: dense Katz-kernel information, repeated pair-specific gain
evaluations, tentative updates, and rollbacks are expensive on large graphs.
Motivated by these bottlenecks, we introduce \algfast{}, the main scalable
algorithm. \algfast{} preserves the same boundary-exchange logic while
replacing dense exact-gain computations with score-based direct-target edge
selection and warm-started iterative Katz updates.
\textbf{Our main contributions are:}
\begin{enumerate}
	\item We introduce the Fair Top-$k$ Katz Centrality Design Problem, a
	minimum-edit graph-design formulation for achieving target group
	representation in Katz top-$k$ rankings and prove strong inapproximability,
	showing that no polynomial-time constant-factor approximation is possible
	unless $\mathrm{P}$=$\mathrm{NP}$.
	
	\item We derive exact Katz sensitivity formulas for directed edge
	additions and show that top-$k$ representation changes are governed by
	boundary gaps between promotable nodes outside the top-$k$ set and
	opposing nodes inside it.
	
	\item We develop \algfast, a scalable boundary-link algorithm that uses
	score-based direct-target batches and warm-started Katz updates instead of
	dense Katz-kernel maintenance.
	
	\item We evaluate \algfast on synthetic and real-world networks, including
	million-node graphs, and show that it achieves the target top-$k$
	representation with few edits while preserving ranking stability.
\end{enumerate}

\noindent
\textbf{Reproducibility:} Code and data are available at
\url{https://github.com/Lutzoe/blade}.

\section{Related Work}

Fairness in rankings and networked systems can be pursued at several different
levels: by post-processing or constraining the ranking
output~\cite{zehlike2017,zehlike2022,celis2017ranking}, by modifying the
scoring or centrality mechanism or reweighting transitions or
edges~\cite{tsioutsiouliklis2021,wang2026fairness,khajehnejad2022crosswalk,
rahman2019fairwalk}, or by intervening directly on the graph
structure~\cite{tsioutsiouliklis2022,masrour2020,liu2024promoting,
jalali2020information,neuhauser2023improving}. Our work belongs to the last
category, but differs from prior graph interventions in targeting the discrete
group composition of a Katz top-$k$ ranking with a minimum-cardinality set of
 edge additions.

\textbf{Fair top-$k$ ranking.}
A line of work studies fairness in ranked lists without explicit graph
structure~\cite{zehlike2017,zehlike2022,celis2017ranking,biega2018equity,
singh2018fairness}. FA*IR and its extensions enforce group-fairness constraints
on prefixes of a top-$k$ ranking~\cite{zehlike2017,zehlike2022}, while related
work studies fair exposure, attention, or scoring rules for top-$k$
selection~\cite{celis2017ranking,biega2018equity,singh2018fairness,
cai2025finding,liu2024a}. These methods typically act on fixed score vectors,
learned scoring functions, or the ranking output itself. In contrast, our
ranking rule remains fixed and fairness is pursued by changing the graph
inducing the scores.

\textbf{Fairness in graph centrality and PageRank.}
Fairness for link-analysis and centrality measures has mainly been studied for
PageRank and related random-walk scores~\cite{tsioutsiouliklis2021,
tsioutsiouliklis2022,wang2026fairness,stoica2024fairness,saxena2024fairsna}.
Existing approaches modify the centrality operator or teleportation
distribution, enforce locally fair transitions, reweight existing edges, or
alter the graph itself through link recommendations or
rewiring~\cite{tsioutsiouliklis2021,tsioutsiouliklis2022,wang2026fairness,
liu2026efficient}. Related work also shows that homophily, group mixing,
observation bias, and network growth can affect minority representation among
highly ranked nodes~\cite{karimi2018homophily,neuhauser2021simulating,espin2022inequality,oliveira2022group,neuhauser2023improving,
shen2025minority}. Our objective instead concerns the discrete composition of
the top-$k$ set under Katz centrality and asks for the minimum number of
admissible edits needed to reach a prescribed target composition.

\textbf{Fair graph construction and enhancement.}
Graph-level interventions have also been used to improve fairness in downstream
tasks such as clustering, link prediction, recommendation, representation
learning, and information flow~\cite{kleindessner2020notion,masrour2020,
liu2024promoting,rahman2019fairwalk,khajehnejad2022crosswalk,jalali2020information,jalali2023fairness}. A related network-design
literature studies graph modifications that improve the centrality or ranking
position of designated nodes or groups~\cite{bergamini2018improving,medya2018group}. These methods demonstrate that
structural changes can alter centrality or mitigate unfairness, but their
objectives differ from achieving a prescribed group composition in a Katz centrality-based top-$k$ set.

\section{Preliminaries and Problem Definition}

Let $G=(V,E)$ be a directed graph with $n=|V|$ nodes.
We assume the nodes are partitioned into two groups $V_r$ (red nodes) and $V_b$ (blue nodes).
We treat $V_r$ as the protected group whose top-$k$ representation is being monitored, 
and $V_b=V\setminus V_r$ is the complementary group.
Let $A \in \mathbb{R}^{n \times n}$ denote the adjacency matrix of $G$.
We write $\mathrm{spec}(A)$ for the \emph{spectrum} of $A$, i.e., the set of its eigenvalues, and define
$
\rho(A) := \max\{|\lambda| : \lambda \in \mathrm{spec}(A)\}
$
as the \emph{spectral radius} of $A$, which is the largest absolute value among the eigenvalues of $A$.\looseness=-1

\begin{definition}
	For a parameter $\alpha \in (0, 1/\rho(A))$, the \emph{Katz kernel} of $G$ is
	$
	K_G = (I - \alpha A)^{-1}.
	$
	Equivalently, it admits the convergent Neumann-series expansion
	\[
	K_G
	=
	\sum_{k=0}^{\infty} \alpha^k A^k,
	\qquad
	\text{for } \alpha < 1/\rho(A).
	\]
\end{definition}

The $(v,a)$ entry of $K_G$ aggregates all directed walks from $v$ to $a$,
with walks of length $k$ discounted by $\alpha^k$.
Thus Katz centrality captures influence propagation as the accumulation of attenuated walks of increasing length~\cite{katz1953new}.
The \emph{Katz centrality score} for each node $a\in V$ is then
\[
s_G(a) := \sum_{v\in V} K_G(v,a).
\]
Equivalently, let $s_G := K_G^\top \mathbf{1}\in\mathbb{R}^n$ denote the vector of Katz centrality scores, so that $s_G(a)=[\,s_G\,]_a$.
Let $T_k(G)\subseteq V$ denote the set of nodes with the $k$ highest centrality scores.\footnote{We define $T_k(G)$ as the set of exactly $k$ nodes with the highest Katz centrality scores. To ensure $|T_k(G)| = k$, we assume that ties in centrality scores are broken deterministically. If equal Katz scores give multiple valid top-$k$ sets, we choose one
	that minimizes $\Unfair(T_k(G))$; any remaining ties are broken by a
	fixed node ordering.}
Let $p = |T_k(G) \cap V_r| / k$ denote the realized proportion of the protected group in the top-$k$, and let $\pi \in [0,1]$ denote the target proportion.
We measure unfairness as the squared deviation from this target
\[
\mathrm{Unfair}(T_k(G))
:=
(p - \pi)^2.
\]
This objective is zero exactly when the target representation is achieved, is symmetric for over- and under-representation, and penalizes larger deviations more strongly. 

\subsection{Graph Design Model}
We focus on unit-cost directed edge additions. Let
$\mathcal{A} \subseteq (V\times V)\setminus E$ denote the set
of admissible new directed edges. A design
$\Delta \subseteq \mathcal{A}$ is a set of added edges, and we
write
\[
G \oplus \Delta := (V, E \cup \Delta).
\]
Given a target tolerance $\varepsilon \in [0,1)$, our goal is to reach
the desired top-$k$ fairness level using as few edge additions as
possible:
\begin{equation}
	\label{eq:design-unit}
	\OPT_\varepsilon
	:=
	\min_{\Delta \subseteq \mathcal{A}} \ |\Delta|
	\quad \text{s.t.}\quad
	\Unfair\bigl(T_k(G \oplus \Delta)\bigr) \le \varepsilon .
\end{equation}
We refer to
\eqref{eq:design-unit} as the \emph{Fair Top-$k$ Katz Centrality
	Design Problem}.

The Neumann-series interpretation of Katz centrality requires $\alpha<1/\rho(A)$. 
In ranking applications, larger or near-critical Katz parameters can be used as empirical ranking parameters; see \cite{aprahamian2016matching} for a discussion.
In the following, we work in the Katz-valid regime throughout the
edit process and assume that
\[
\alpha < \frac{1}{\rho(A_{G\oplus\Delta})}
\]
for every edit set $\Delta \subseteq \mathcal{A}$ considered by the algorithm.
This ensures that every Katz kernel used below admits the
nonnegative Neumann-series representation.

\subsection{Hardness}\label{sec:hardness}

We first show that the target-cost formulation is not only hard to
solve exactly, but also hard to approximate. Let $\OPT_\varepsilon=\infty$ if no feasible edit set exists.

\begin{theorem}\label{thm:dual-inapprox-katz}
	For every constant $\psi \ge 1$, the Fair Top-$k$ Katz Centrality
	Design Problem admits no polynomial-time $\psi$-approximation unless
	$\mathrm{P}=\mathrm{NP}$, even when $\alpha=\tfrac12$ and the input
	graph is a directed acyclic graph, i.e., no polynomial-time
	algorithm can always return a feasible edit set $\Delta$ satisfying
	$
	|\Delta|\le \psi\,\OPT_\varepsilon .
	$
\end{theorem}

\begin{proof}[Proof idea]
	The proof is by a gap-preserving reduction from
	\textsc{Independent Set} on $3$-regular graphs. 
	Given an instance
	$(H,\ell)$, we construct a Fair Top-$k$ Katz Centrality Design
	instance such that, in the YES case, the fairness target can be
	reached using at most $\ell$ edge additions, whereas in the NO case
	every feasible solution requires at least $Q$ edge additions, for an
	integer $Q>\psi\ell$. Therefore, a polynomial-time
	$\psi$-approximation would distinguish the two cases and solve
	\textsc{Independent Set} in polynomial time.
	The full proof is provided in \Cref{app:proofs}.
\end{proof}

The preceding theorem immediately rules out an exact polynomial-time
algorithm. It also implies hardness of the associated threshold decision
problem.

\begin{corollary}\label{cor:decision-nphard-katz}
	The following decision problem is NP-hard: given an instance,
	a tolerance $\varepsilon$, and an integer threshold $C$, decide
	whether there exists an edit set $\Delta\subseteq\mathcal A$ with
	$|\Delta|\le C$ such that
	$
	\Unfair\bigl(T_k(G\oplus\Delta)\bigr)\le \varepsilon .
	$
\end{corollary}

\section{Structural Properties of the Katz Kernel}
\label{sec:structural}

We derive structural properties of the Katz kernel that underpin both
algorithms developed in~\Cref{sec:alg,sec:fast}.
\begin{proposition}
	\label{prop:monotone}
	Let $G' = G \oplus (u \to v)$ be obtained by adding a directed edge.
	If $\alpha < 1/\rho(A_{G'})$, then
	\[
	K_{G'}(a,b) \ge K_G(a,b)
	\quad
	\text{for all } a,b \in V.
	\]
\end{proposition}

The rank-one structure of directed edge additions yields closed-form
marginal updates under a single edit and lets us score candidate edges
without recomputing $(I-\alpha A)^{-1}$ from scratch.

\begin{lemma}
	\label{lem:katz-sensitivity}
	Let $K_G = (I - \alpha A_G)^{-1}$ and let
	$G' = G \oplus (u \to v)$. Then
	\begin{equation}
		\label{eq:katz-update}
		K_{G'}
		=
		K_G
		+
		\frac{\alpha}{1 - \alpha K_G(v,u)}
		\mathbf{k}_{\cdot u}\mathbf{k}_{v\cdot},
	\end{equation}
	where $\mathbf{k}_{\cdot u}=K_G\mathbf{e}_u$ is the $u$-th column of
	$K_G$ and $\mathbf{k}_{v\cdot}=\mathbf{e}_v^\top K_G$ is the $v$-th
	row of $K_G$. Moreover, the centrality change at node $a$ from
	adding edge $u\to v$ is
	\begin{equation}
		\label{eq:score-change}
		\Delta s(a\mid u\to v)
		=
		\frac{\alpha\, s_G(u)\, K_G(v,a)}
		{1-\alpha K_G(v,u)} .
	\end{equation}
\end{lemma}

The boost to $a$ is the product of three factors:
(i) $s_G(u)$, the importance of the tail node;
(ii) $K_G(v,a)$, the kernel proximity from the head to the target; and
(iii) $(1-\alpha K_G(v,u))^{-1}$, a feedback amplification term that is
large when $v$ already has high influence on $u$. If our goal is to
promote node $a$ in the ranking, it is not enough that $a$'s score
increases; it must increase \emph{more than} the score of any node it
needs to overtake. This motivates working with differential rather than
absolute gains.

\begin{definition}
	\label{def:diff-gain}
	For a candidate pair $(a,b)$ and directed edge $e=(u\to v)$, the
	\emph{differential gain} is
	\begin{equation}
		\delta_G(a,b\mid u\to v)
		:=
		\Delta s(a\mid u\to v)-\Delta s(b\mid u\to v).
	\end{equation}
	The \emph{nonnegative score gap} is
	\[
	\Gamma_G(a,b):=\max\{0, s_G(b)-s_G(a)\}.
	\]
\end{definition}

After adding edge $e$, the score difference becomes
\[
s_{G\oplus e}(b)-s_{G\oplus e}(a)
=
(s_G(b)-s_G(a))-\delta_G(a,b\mid e).
\]
Thus, when $\Gamma_G(a,b)>0$, a positive differential gain reduces the
amount by which $a$ trails $b$; accumulated gain at least
$\Gamma_G(a,b)$ closes the current pairwise gap.

\begin{proposition}
	\label{prop:factorization}
	The following factorization holds:
	\begin{equation}
		\label{eq:factored}
		\delta_G(a,b\mid u\to v)
		=
		\underbrace{
			\frac{\alpha\,s_G(u)}
			{1-\alpha K_G(v,u)}
		}_{\displaystyle \lambda_G(u,v)>0}
		\cdot
		\underbrace{
			\bigl(K_G(v,a)-K_G(v,b)\bigr)
		}_{\displaystyle \mu_G(v;\,a,b)} .
	\end{equation}
	Hence the differential gain factors into a positive edge-dependent
	scalar $\lambda_G(u,v)$ and a \emph{kernel gap}
	$\mu_G(v;\,a,b)$ that depends only on the head node $v$ and the
	target pair.
\end{proposition}

The factorization clarifies the distinct roles of head and tail. The
sign of $\delta_G$ is determined entirely by the head $v$ through
$\mu_G(v;\,a,b)$: an edge helps pair $(a,b)$ if and only if the total
weight of $\alpha$-discounted walks from $v$ to $a$ is larger than that
from $v$ to $b$. The
magnitude is controlled by $\lambda_G(u,v)$, which depends on the tail
$u$ through $s_G(u)$ and on the feedback amplification through
$(1-\alpha K_G(v,u))^{-1}$. Consequently,
\[
\delta_G(a,b\mid u\to v)>0
\quad\Longleftrightarrow\quad
K_G(v,a)>K_G(v,b).
\]
The tail node $u$ affects magnitude but not sign.

\section{Frontier-Guided Algorithms}
\label{sec:algorithms}
By \Cref{thm:dual-inapprox-katz}, no polynomial-time algorithm can provide approximation guarantees for the  Fair Top-$k$ Katz Centrality
Design Problem unless $\mathrm{P}=\mathrm{NP}$. We therefore target practical efficiency rather than worst-case certificates and 
develop frontier-guided algorithms.
The key observation is that the fairness objective
changes only when nodes cross the boundary of the top-$k$ set. Thus, rather
than trying to reshape all Katz scores globally, we focus on pairs of nodes
whose relative order can change the group composition of the top-$k$ set.
Let $H$ be the current graph and let
\[
h(H):=|T_k(H)\cap V_r|
\]
denote the number of protected nodes in the top-$k$. We measure progress by
the per-count objective
\[
\Phi(c):=\Bigl(\frac{c}{k}-\pi\Bigr)^2.
\]
The target is reached exactly when
\[
\Phi(h(H))\le \varepsilon,
\qquad\text{equivalently,}\qquad
\Unfair(T_k(H))\le \varepsilon.
\]

If $h(H)/k<\pi$, the protected group is underrepresented, and the algorithm
tries to promote protected nodes into the top-$k$. If $h(H)/k>\pi$, the
roles are reversed, and the algorithm promotes non-protected nodes. Let $P$
be the currently promoted group and $O=V\setminus P$ the opposing group.

\subsection{Top-\texorpdfstring{$k$}{k} Frontier Principle}\label{sec:frontier}

The active frontier is
\[
\mathcal F(H):=(P\setminus T_k(H))\times (O\cap T_k(H)).
\]
A pair $(a,b)\in\mathcal F(H)$ consists of a promoted node $a$ outside the
top-$k$ and an opposing node $b$ inside the top-$k$. Its current
nonnegative score gap is $\Gamma_H(a,b)$.
Closing this gap makes $a$ competitive with an opposing node currently in
the top-$k$, and is therefore the basic local operation used by our
algorithms.

For an optional frontier-window parameter $L\in\mathbb N\cup\{\infty\}$, we
define $\mathcal F_L(H)\subseteq \mathcal F(H)$ by ordering
$P\setminus T_k(H)$ by decreasing Katz score, ordering $O\cap T_k(H)$ by
increasing Katz score, breaking ties by a fixed node order, and taking the
first $L$ pairs in the resulting lexicographic product. If $L=\infty$,
then the full frontier is searched. 

For an admissible edge $e$, define the differential gain of $e$ for a
frontier pair $(a,b)$ as
\[
\delta_H(a,b\mid e)
=
\bigl(s_{H\oplus e}(a)-s_H(a)\bigr)
-
\bigl(s_{H\oplus e}(b)-s_H(b)\bigr).
\]
By \Cref{prop:factorization}, this quantity admits a closed-form Katz
factorization. Positive differential gain means that the edge helps the
promoted node $a$ more than the opposing node $b$.

Let
$
g_1^H(a,b)\ge g_2^H(a,b)\ge \cdots>0
$
be the positive differential gains over all currently admissible edges,
sorted in nonincreasing order. We define the cost of pair
$(a,b)$ as
\[
C_H(a,b)
:=
\min\left\{
r:
\sum_{i=1}^r g_i^H(a,b)\ge \Gamma_H(a,b)
\right\},
\]
with $C_H(a,b)=0$ if $\Gamma_H(a,b)=0$, and $C_H(a,b)=\infty$ if no such
$r$ exists. Intuitively, $C_H(a,b)$ is the number of currently strongest
helpful edits needed to close the present pair gap.

\subsection{Dense Frontier Search}\label{sec:alg}

As a reference method, \alg searches the frontier
$\mathcal F_L(G)$ using the exact differential gains from
\Cref{lem:katz-sensitivity}. For each frontier pair $(a,b)$, it greedily
examines positive-gain admissible edges and evaluates tentative prefixes
using the rank-one Katz update. A prefix is committed only if it strictly
decreases $\Phi(h(G))$, or leaves $\Phi(h(G))$ unchanged while strictly
reducing $\Gamma_G(a,b)$; otherwise, the tentative edits are rolled back.
Hence, under Katz validity, every committed round weakly decreases the
fairness objective, while plateau commits make progress on the targeted
boundary gap.

This procedure requires dense Katz-kernel maintenance, pair-specific gain
evaluations, and repeated tentative updates and rollbacks, making it
suitable only as a small-graph reference. \algfast retains the
boundary-exchange principle while replacing exact kernel-based gains with
a scalable score proxy.

\subsection{\algfast: Scalable Boundary-Link Search}
\label{sec:fast}

We now present \algfast{} (\textbf{B}oundary-\textbf{L}ink \textbf{A}ugmentation by
\textbf{D}irect \textbf{E}dges), a scalable boundary-link algorithm (\Cref{alg:fast}). The method keeps
the central principle of the boundary: fairness changes only when nodes cross
the top-$k$ boundary, so edits should be spent on promotable nodes that are
already close to entering the top-$k$. In contrast to the dense algorithm, however, \algfast
does not maintain the dense Katz kernel and does not evaluate all pair-specific
differential gains. It uses only Katz scores, admissible direct-target edges,
and Jacobi score updates after each committed batch.

\textbf{Setup.}
At a current graph $G$, let $P$ be the group that should be promoted and
let $O=V\setminus P$ be the opposing group. Define the vulnerable set
$
D:=O\cap T_k(G).
$
\algfast anchors the gap estimate on a single \emph{boundary node}
\[
b^\star := \argmin_{b \in D}\ s_G(b),
\]
i.e., the lowest-scoring opposing node currently in the top-$k$. 
Focusing on the weakest opposing node is consistent with the boundary-exchange
principle: if a promotable node is to improve the top-$k$ group composition,
it must eventually displace an opposing node, and the easiest such displacement
is against the lowest-scoring opposing node $b^\star$.

\textbf{Batch estimation.}
For a parameter $q\in\mathbb N$, \algfast considers the boundary candidate
set $\mathcal C_q(G)$ consisting of the $q$ highest-scoring nodes in
$P\setminus T_k(G)$ that have at least one remaining admissible incoming edge. If fewer than $q$ such nodes exist, all
of them are considered.
For each candidate $a\in\mathcal C_q(G)$, \algfast constructs a \emph{source
	batch} by greedily accumulating admissible tails in decreasing order of their
current Katz score. Let $\sigma_1,\sigma_2,\ldots$ be the nodes of $V$
sorted so that $s_G(\sigma_1)\ge s_G(\sigma_2)\ge \cdots$, with ties broken
by fixed node order. The batch $\mathcal B(a)$ is built by scanning
$\sigma_1,\sigma_2,\ldots$ and appending $\sigma_i\to a$ whenever
$\sigma_i\neq a$ and $\sigma_i\to a \in \mathcal{A}$, until the cumulative estimated gain closes $\Gamma_G(a,b^\star)$, adding each source in turn and subtracting the
proxy $\alpha\cdot s_G(\sigma_i)$ from the remaining gap. The batch is
accepted as soon as the remaining gap reaches zero. 

The gain proxy uses $s_G(u)$ as a surrogate for the direct-target
differential gain. By \Cref{eq:score-change}, adding edge $u\to a$ boosts
$s_G(a)$ by
$
\tfrac{\alpha\, s_G(u)\, K_G(a,a)}{1-\alpha K_G(a,u)},
$
so high-score tails produce large score increases at $a$. Neglecting the
kernel factors $K_G(a,a)$ and $(1-\alpha K_G(a,u))^{-1}$, which are not
maintained by \algfast, the dominant factor is $\alpha\,s_G(u)$, which is
the proxy used in the estimation.

Among all candidates $a\in\mathcal C_q(G)$ for which a valid batch
$\mathcal B(a)$ exists, \algfast selects the best candidate by preferring
the smallest batch size (fewest edits), breaking ties by the highest current
score of $a$, and finally by the fixed node order. The edges of the winning
batch are then committed one by one to $G$, in the same greedy source order.
If the scan ends before the estimated gap closes, the candidate $a$ is discarded.

\textbf{Score update.}
After committing the batch, \algfast updates the Katz scores using warm-started Jacobi iterations,
$
s \leftarrow \mathbf{1} + \alpha A_G^\top s,
$
warm-started from the current score vector, until convergence.
It then updates the top-$k$ set and
group counts accordingly. The loop repeats until the fairness target is
reached, or no valid candidate exists.

\begin{figure}
	\centering
	\includegraphics[width=0.7\columnwidth]{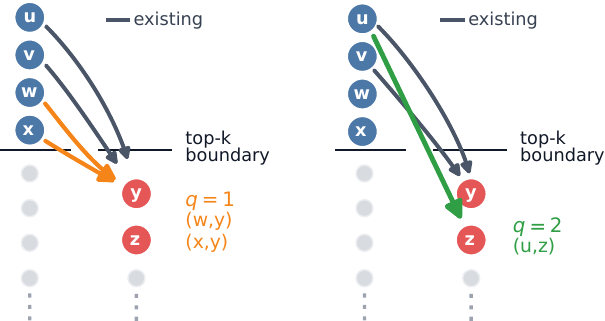}
	\caption{A case with $k=4$ where $q=1$ is suboptimal. Initially,
		$T_4=\{u,v,w,x\}$; $y$ and $z$ are the two highest-scoring actionable
		red nodes outside $T_4$. Promotion means entering $T_4$ and displacing
		a blue node. With $q=1$, only $y$ is evaluated and needs two edges;
		with $q=2$, $z$ is also evaluated and $(u,z)$ suffices.}
	\Description{Two graph panels compare boundary widths one and two. The first requires two added edges to promote node y, while the second identifies node z, which needs one edge.}
	\label{fig:blade-q-counterexample}
\end{figure}

\begin{algorithm}[t]
	\caption{\algfast}
	\label{alg:fast}
	\small
	\KwIn{Directed graph $G=(V,E)$, admissible edges $\mathcal A$,
		top-$k$ size $k$, target proportion $\pi$, tolerance $\varepsilon$,
		Katz parameter $\alpha$, boundary width $q$}
	\KwOut{Edit set $\Delta$}
	
	Compute Katz scores $s_G$ and top-$k$ set $T_k(G)$\;
	$\Delta\gets\emptyset$\;
	
	\While{$\Unfair(T_k(G))>\varepsilon$}{
		Determine promoted group $P$ and opposing group $O$\;
		$D\gets O\cap T_k(G)$\;
		\lIf{$(P\setminus T_k(G))=\emptyset$ \textbf{or} $D=\emptyset$}{
			\Break
		}
		
		$b^\star\gets \argmin_{b\in D} s_G(b)$\;
		Let $\mathcal C_q(G)$ be the $q$ highest-scoring actionable nodes in
		$P\setminus T_k(G)$\;
		\lIf{$\mathcal C_q(G)=\emptyset$}{
			\Break
		}
		
		Sort all nodes as $\sigma_1,\ldots,\sigma_n$ by decreasing $s_G$\;
		$\mathit{best}\gets\emptyset$\;
		
		\ForEach{$a\in\mathcal C_q(G)$}{
			$\Gamma\gets \Gamma_G(a,b^\star)$\;
			$\mathcal B(a)\gets\emptyset$\;
			
			\ForEach{$\sigma_i\in\sigma_1,\ldots,\sigma_n$}{
				\lIf{$\sigma_i=a$ \textbf{or} $(\sigma_i\to a)\notin\mathcal A$}{
					\Continue
				}
				
				$\mathcal B(a)\gets\mathcal B(a)\cup\{\sigma_i\to a\}$\;
				$\Gamma\gets\Gamma-\alpha\cdot s_G(\sigma_i)$\;
				
				\If{$\Gamma\le 0$}{
					\If{$\mathit{best}=\emptyset$
						\textbf{or} $|\mathcal B(a)|<|\mathit{best}|$
						\textbf{or}
						$\bigl(|\mathcal B(a)|=|\mathit{best}|
						\textbf{ and } s_G(a)>s_G(a_{\rm best})\bigr)$}{
						$\mathit{best}\gets\mathcal B(a)$\;
						$a_{\rm best}\gets a$\;
					}
					\Break
				}
			}
		}
		
		\lIf{$\mathit{best}=\emptyset$}{
			\Break
		}
		
		\tcp{Commit selected batch before updating scores.}
		\ForEach{edge $e\in\mathit{best}$ in greedy source order}{
			Add $e$ to $G$; $\Delta\gets\Delta\cup\{e\}$;
			$\mathcal A\gets\mathcal A\setminus\{e\}$\;
		}
		
		Update Katz scores $s_G$\; 
		Recompute $T_k(G)$ and group counts using $s_G$\;
	}
	
	\If{$\Unfair(T_k(G))\le\varepsilon$}{
		\Return $\Delta$, succeed\;
	}
	\Return $\Delta$, failed\;
\end{algorithm}

\textbf{Proxy interpretation.}
\algfast estimates the benefit of a direct-target edge \(u\to a\) by the proxy
$
\widehat g_G(u,a):=\alpha s_G(u).
$
This proxy captures the tail-node contribution to the Katz update, but ignores
kernel amplification and the possible spillover to the opposing boundary node
\(b^\star\).

\begin{proposition}\label{prop:maxincrease}
	Let \(u\to a\) be an admissible direct-target edge with Katz validity after
	adding this edge. Then
	\[
	\Delta s_G(a\mid u\to a)
	=
	\widehat g_G(u,a)\,
	\frac{K_G(a,a)}{1-\alpha K_G(a,u)}
	\]
	and, for the opposing boundary node \(b^\star\),
	\[
	\delta_G(a,b^\star\mid u\to a)
	=
	\widehat g_G(u,a)\,
	\frac{K_G(a,a)-K_G(a,b^\star)}
	{1-\alpha K_G(a,u)} .
	\]
	Consequently, \(\widehat g_G(u,a)\) never overestimates the direct score
	increase of \(a\). However, it may overestimate or underestimate the
	differential gain against \(b^\star\). In particular,
	\[
	\delta_G(a,b^\star\mid u\to a)>0
	\quad\Longleftrightarrow\quad
	K_G(a,a)>K_G(a,b^\star).
	\]
\end{proposition}

\textbf{Boundary search width.}
The parameter $q$ controls the width of the boundary search. For $q=1$, \algfast considers only the single highest-scoring actionable node in $P \setminus T_k(G)$ as the promotion candidate. This is the node closest to the top-$k$ boundary, but it need not be the most efficiently promotable one: the gap it must close may be large, or the available sources may be weak, requiring many edges.
For $q>1$, \algfast evaluates the $q$ highest-scoring candidates before committing, selecting the one whose batch is smallest. Figure~\ref{fig:blade-q-counterexample} illustrates why this matters. In the left panel ($q=1$), only boundary node $y$ is considered and since $w$ and $x$ are the best available sources into node $y$, \algfast selects edges $(w,y)$ and $(x,y)$, spending two edits to promote $y$'s score. In the right panel ($q=2$), node $z$ is also considered as a boundary node. Here, a single edge $(u,z)$ gives $z$ enough support to cross the top-$k$ boundary. \algfast commits this one-edit solution instead, reaching the target at half the cost. The gain from larger $q$ is thus bounded by the cost difference between the best and worst near-boundary candidates; in practice, small values $q \in \{2,3\}$ already recover most of this benefit.

\subsection{Complexity}
Let $n=|V|$, $m=|E|$, and let $M=|\mathcal A|$ denote the
number of admissible edge additions. 
Let $\Lambda:=\max_t |\mathcal F_L(G_t)|$ be the maximum number of
frontier pairs searched in one round; $\Lambda\le L$ for finite $L$,
and $\Lambda\le k(n-k)$ for the full frontier.

The dense frontier algorithm (\Cref{sec:alg}) requires dense Katz-kernel maintenance. The initial kernel
computation costs $O(n^3)$ time and $O(n^2)$ space. For each searched
frontier pair, evaluating all admissible differential gains costs
$O(M)$, or $O(M\log M)$ if the gains are sorted, and each tentative
rank-one update costs $O(n^2)$. Space is dominated by
$O(n^2+m)$ for kernel maintenance. Hence the dense approach is mainly a reference
method for smaller graphs.

For \algfast, sorting nodes by score costs $O(n\log n)$ per round. Selecting the top-$q$ actionable boundary candidates costs $O(n\log q)$, and scanning the sorted source list for these candidates costs $O(qn)$, assuming constant-time admissibility checks. After committing the chosen batch, \algfast recomputes Katz scores by warm-started Jacobi iterations. If $I_t$ denotes the number of iterations in round $t$, this update costs $O(I_t(m+n))$. Thus, one outer round costs
$O(n\log n+n\log q+qn+I_t(m+n)).$
Overall, \algfast avoids dense Katz-kernel maintenance, tentative rollbacks, and scans over all admissible edges. Its space complexity is $O(n+m)$.

Finally, when $\mathcal A$ is the unconstrained set of all missing directed non-self-loop edges, it does not need to be materialized. In this case, admissibility of a candidate edge can be tested by checking whether the edge is currently absent from the graph.

\begin{table}[t]
	\centering
	\caption{Summary of datasets.} 
	\label{tab:datasets}
	\resizebox{1\linewidth}{!}{		\begin{tabular}{lrrrcccl}
			\toprule
			\textbf{Dataset} &
			$\boldsymbol{|V|}$ &
			$\boldsymbol{|E|}$ &
			\textbf{Avg. deg.} &
			\textbf{Classes} &
			\textbf{Minor.} &
			\textbf{Domain} &
			\textbf{Ref.} \\
			\midrule
			\emph{Blogs} & $1\,224$ & $19\,022$ & $31.1$ & Left/Right  & $48\%$ & Political & \cite{adamic2005political} \\
			\emph{Hopkins} & $5\,180$ & $186\,586$ & $72.0$ & Fem./Male & $45\%$ & Social & \cite{traud2012social} \\
			\emph{Retweet} & $18\,470$ & $48\,365$ & $5.2$ & Left/Right  & $38\%$ & Political & \cite{conover2012partisan} \\
			\emph{Deezer} & $28\,281$ & $92\,752$ & $6.6$ & Fem./Male & $44\%$ & Social & \cite{rozemberczki2020characteristic} \\
			\emph{Penn} & $41\,554$ & $1\,362\,229$ & $65.6$ & Fem./Male & $48\%$ & Social & \cite{traud2012social} \\
			\emph{Pokec} & $1\,632\,803$ & $30\,622\,564$ & $37.5$ & Fem./Male & $49\%$ & Social & \cite{takac2012data} \\
			\bottomrule
	\end{tabular}}
\end{table}

\section{Experiments}
We discuss the following research questions:

\begin{itemize}
	\item \textbf{RQ1 Effectiveness and edit efficiency:} Can the proposed edge-addition methods achieve the target Katz top-$k$ group representation with fewer edits than baselines?
	\item \textbf{RQ2 Scalability:} How efficient are our algorithms?
	\item \textbf{RQ3 Sensitivity to graph structure and parameters:} How do the Katz parameter, top-$k$ cutoff, target proportion, and boundary width affect edit cost, and how much does \algfast perturb the original Katz ranking?
\end{itemize}

\subsection{Experimental Setup}
\textbf{Algorithms.}
Since no prior method directly targets fair top-$k$ composition under Katz centrality, 
we introduce natural baselines spanning three strategies: exact and greedy optimization, 
global mass balancing, and group-based support. 
To further contextualize the problem relative to existing link-analysis fairness methods, 
we report PageRank-based fairness methods under their native PageRank rankings. 
Specifically, we use:\looseness=-1

\begin{itemize}
	\item \blopt: An exact solver that performs depth-wise enumeration over
	admissible edge subsets in nondecreasing cardinality. It is used only on
	the small synthetic instances.
	
	\item \blgapgreedy: At each iteration, it determines the promoted group
	$P$ and opposing group $O$ and let $b^\star$ be the lowest-scoring node in
	$O\cap T_k(G)$, and considers admissible edges into boundary candidates
	$a \in P\setminus T_k(G)$. For each candidate edge $e$,
	it computes its exact one-edge reduction of the boundary gap
	$\Gamma_G(a,b^\star)$. It then adds the edge with largest positive gap reduction,
	recomputes Katz scores, and repeats until the target is reached.
	
 \item \blsamegroup:
Boosts underrepresented nodes by adding support edges among nodes of the same group. At each step it considers a window of the top-ranked underrepresented nodes outside the top-$k$
(window size $W=100$), forms batches of within-group support edges around them, commits the batch that most improves the top-$k$ composition without worsening
it, recomputes Katz scores, and stops once the target is reached.

\item \blkatzmass: For a graph $G$, let
$
M_r(G):=
\frac{\sum_{v\in V_r} s_G(v)}
{\sum_{v\in V} s_G(v)}
$
be the protected group's share of total Katz centrality mass, and define
$\Psi_{\mathrm{mass}}(G):=(M_r(G)-\pi)^2$. At each step, \blkatzmass samples admissible edges whose endpoints move Katz mass in the direction that reduces the protected-group mass
deviation. It uses $\lceil 0.001|E|\rceil$ random edge attempts per score-update commit,
recomputes the top-$k$ ranking after each committed batch, and stops when either the desired top-$k$ composition or the $\Psi_{\mathrm{mass}}$ target is reached.

	\item \alg: Our dense algorithm described in
	\Cref{sec:alg}. Given the $O(n^3)$ runtime it is unusable in practice; it functions purely as a small-scale sanity reference.

	\item \algfast: Our scalable algorithm from
	\Cref{sec:fast}. We set the boundary width $q=2$ unless stated otherwise and report exact converged Katz results after \algfast edits. For the Katz convergence, we use a tolerance of $1e-10$ and at most 200 iterations.

\end{itemize}
We additionally use fair-PageRank (PR) algorithms because no established fairness baselines exist for Katz-centrality top-$k$ design.
\textsc{FairWalk}~\cite{rahman2019fairwalk} balances random-walk transition probabilities across groups,
\textsc{CrossWalk}~\cite{khajehnejad2022crosswalk} reweights walks to improve fairness in graph representation learning, and \textsc{FairGD}~\cite{wang2026fairness} reweights PageRank transitions.
\bllfprn and \bllfpru~\cite{tsioutsiouliklis2021} impose locally fair
PageRank transitions using neighborhood-based and uniform residual
allocation, respectively.

\textbf{Experimental settings.}
Unless otherwise stated, all real-network experiments use $k=100$, target proportion $\pi=1/2$, and tolerance $\varepsilon=0$, i.e., the target is exact parity in the top-$100$ set with 50 protected and 50 non-protected nodes.
We use $\pi=1/2$ only as a controlled benchmark. The framework supports any
stakeholder-specified target $\pi$, whose application-dependent choice is
outside our scope; RQ3 evaluates sensitivity to $\pi$.

We set $\alpha=1/d_{\max}$, where $d_{\max}$ is the maximum out-degree of the original directed graph. This is conservative: across all datasets, $\alpha$ is below both $1/\rho(A)$ and $1/\rho(A+\Delta)$ after intervention; in our runs, the largest observed ratio $\alpha/(1/\rho(A+\Delta))=\alpha\rho(A+\Delta)$ was $0.17$ on Hopkins, and all other datasets had larger margins. 
The admissible edge set $\mathcal{A}$ contains all missing directed
non-self-loop edges. We use this unconstrained set as an algorithmic benchmark:
it gives all edge-addition methods the same design space and measures what they can achieve
without conflating optimization quality with a particular application's
eligibility rules. It is not intended to imply that every missing edge is a
valid intervention in deployment. We therefore also evaluate a local two-hop
admissible set, where an edge $u \to v$ is eligible only if $u$ can reach $v$
by a directed path of length two in the original graph. This restriction
serves as a structural proxy for locally plausible introductions or
friend-of-friend recommendations;
a deployed system should further filter $\mathcal A$ using consent, relevance,
safety, and other domain requirements.

All methods are run with a one-hour time limit.
All experiments were conducted on a CPU-only cluster, with each run allocated a dedicated compute node equipped with an AMD EPYC 9634 processor and 1.5 TiB of RAM.

\textbf{Datasets.}
We evaluate both synthetic and real-world networks. 
The synthetic benchmarks are biased preferential attachment (BPA) graphs with controlled group imbalance and homophily, using the BPA model introduced by Avin et al.~\cite{avin2015homophily}. 
\Cref{tab:datasets} gives an overview of the real-world datasets. We use the node classes released with the original datasets.

\subsection{Results}

\begin{figure}[t]
	\centering
	\includegraphics[width=0.9\linewidth]{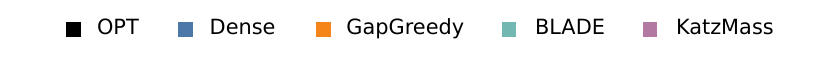}
	\begin{subfigure}{0.5\columnwidth}
		\centering
		\includegraphics[width=.9\linewidth]{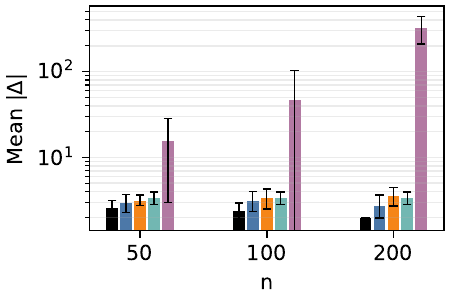}
		\caption{Mean added edges.}
		\label{fig:syn1}
	\end{subfigure}\hfill	\begin{subfigure}{0.5\columnwidth}
		\centering
		\includegraphics[width=.9\linewidth]{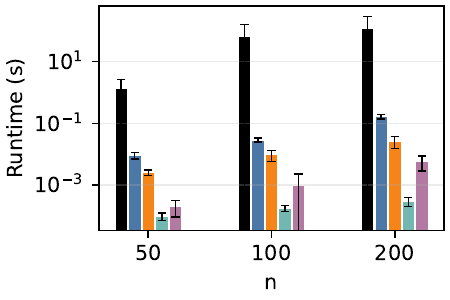}
		\caption{Mean running time.}
		\label{fig:syn2}
	\end{subfigure}
	\caption{Experiments on small synthetic graphs. \blsamegroup not shown as it did not find any feasible solutions.}
	\Description{Bar charts compare the mean number of added edges and mean running time of six methods on synthetic graphs with 50, 100, and 200 nodes.}
	\label{fig:syn}
\end{figure}

\noindent
\underline{\textbf{RQ1 Effectiveness and edit efficiency:}}
\blopt is computationally feasible only on small instances, so we use it as an exact reference in a synthetic scaling study. We generate directed BPA graphs with homophily $\rho=0.5$, attachment parameter $m=2$, minority fraction $0.35$, target $\pi=0.5$, $k=6$, and $\alpha=1/d_{\max}$. For $n \in \{50,100,200\}$, we average over 10 independent runs. All methods use the same two-hop admissible edge set around promotable nodes. 

\Cref{fig:syn} shows that \blopt gives the fewest edits when it finishes, but its runtime grows quickly and it begins to time out as $n$ increases. The dense boundary method stays closest to \blopt in edit count, but also becomes costly. In contrast, \algfast reaches the target on all tested instances, is orders of magnitude faster, and uses only a few additional edits. \blgapgreedy succeeds reliably but is slower than \algfast and does not improve the edit-quality/runtime tradeoff. Thus, boundary-based methods remain close to the exact reference where it is feasible, while \algfast provides the scalable alternative.

\textbf{Real-world graph results.}
\Cref{tab:katz_results,tab:Katz-edits} evaluate our algorithms on real-world networks.
The results show that \algfast is the only method that consistently solves the fair top-$k$ Katz design task on all six real-world networks. For $k=100$, $\pi=1/2$, and $\varepsilon=0$, \algfast reaches zero final unfairness on every dataset, whereas the baselines are substantially less reliable.
\Cref{tab:katz_results} reports the resulting top-$k$ group composition.
Since $k=100$ and the target is 50 protected nodes, \algfast must increase
the protected count by $50$ minus the initial protected count. 
\Cref{tab:Katz-edits} reports the corresponding intervention cost in added edges.

\blgapgreedy is a strong boundary-based baseline and is competitive when it finishes, but unlike \algfast it does not batch edge additions and recomputes Katz scores after each single edit; as a result, it times out on four of the six real-world datasets. 
\blkatzmass often fails to improve and in several cases actively worsens the top-$k$ objective, because increasing aggregate group mass does not constrain which nodes receive the boost: edges that raise the protected group's total score may simultaneously elevate opposing nodes already inside the top-$k$ boundary, confirming that balancing aggregate Katz mass is not sufficient for controlling top-$k$ composition. \blsamegroup is the strongest non-boundary baseline, but it typically requires many more edits and misses exact parity on Penn.

The edit counts further support the boundary-based design principle. Across the datasets where both methods reach or nearly reach the target, \algfast uses substantially fewer edits than \blsamegroup. \blgapgreedy also demonstrates the value of boundary-aware gap closing, but its lack of batching prevents it from scaling. Overall, these results indicate that fair top-$k$ Katz design is best addressed by targeted and batched boundary interventions rather than by global Katz-mass optimization, untargeted group-based edge additions, or unbatched greedy gap closing.

\noindent
\textbf{Local two-hop interventions.}
\Cref{tab:twohop-blade} compares \algfast under the
unconstrained admissible set used in the main experiments and under the
two-hop admissible set. \algfast reaches the target on all six real-world datasets in both
settings. The two-hop constraint changes the edit count only marginally: the
total number of edits increases from $5{,}546$ to $5{,}601$, an increase of
only $55$ edits overall. Three of the six datasets require exactly the same
number of edits.

\noindent
\textbf{PageRank comparison.}
\Cref{tab:k100-table-1-PR-reduction} gives contextual comparisons to fairness-aware PageRank and random-walk methods under their native PageRank rankings. These methods are not direct competitors, since they modify transitions, walks, or representations rather than adding edges to optimize Katz top-$k$ representation. Their effect is highly dataset-dependent: while they reduce top-$k$ PageRank unfairness on some datasets, they can also leave the objective unchanged or substantially worsen unfairness on others. For example, \bllfprn eliminates the PageRank top-$k$ deviation on Blogs, Hopkins, and Penn, but substantially increases it on Retweet. Thus, fairness of global PageRank mass does not ensure fair top-$k$ representation, and PageRank-oriented interventions do not directly solve the Katz top-$k$ graph-design objective studied here.

\begin{table}[t]
	\centering
	
	\caption{Katz-based methods evaluation for $k=100$. OOT denotes timeout after one hour.}
	\label{tab:k100-evaluation-main}
	
	\begin{subtable}[t]{\linewidth}
		\centering\resizebox{\linewidth}{!}{			\begin{tabular}{lrrrrrr}
				\toprule
				\textbf{Algorithm} & Blogs & Hopkins & Retweet & Deezer & Penn  & Pokec \\
				\midrule
				\blgapgreedy & 100.0 & 100.0 & OOT & OOT & OOT  & OOT \\
				\blsamegroup & 100.0 & 100.0 & 100.0 & 100.0 & 97.2 & 100.0 \\
				\blkatzmass & -300.0 & -77.8 & -12.9 & 0.0 & 12.9 &  0.0 \\
				\algfast & 100.0 & 100.0 & 100.0 & 100.0 & 100.0 & 100.0 \\
				\bottomrule
		\end{tabular}}
		\caption{Reduction of top-$k$ unfairness (\%).}
		\label{tab:Katz-reduction}
	\end{subtable}

	\begin{subtable}[t]{\linewidth}
		\centering
		
		\resizebox{\linewidth}{!}{
			\begin{tabular}{lcccccc}
				\toprule
				\textbf{Algorithm} & Blogs & Hopkins & Retweet & Deezer & Penn & Pokec \\
				\midrule
				Initial \# protected  & 48 & 44 & 34 & 26 & 20 & 38 \\
				\midrule
				\blgapgreedy
				& 50 & 50 & OOT & OOT & OOT & OOT \\
				\blsamegroup
				& 50 & 50 & 50 & 50 & 45 & 50 \\
				\blkatzmass
				& 46 & 42 & 33 & 26 & 22 & 38 \\
				\algfast
				& 50 & 50 & 50 & 50 & 50 & 50 \\
				\bottomrule
			\end{tabular}
			
		}
		\subcaption{Numb.~of protected nodes in top-$k$ before and after intervention.}
		\label{tab:katz_results}
	\end{subtable}
	
	\begin{subtable}[t]{\linewidth}
		\centering\resizebox{\linewidth}{!}{			\begin{tabular}{lrrrrrr}
				\toprule
				\textbf{Algorithm} & Blogs & Hopkins & Retweet & Deezer & Penn & Pokec \\
				\midrule
				\blgapgreedy & 5 & 100 & OOT & OOT & OOT & OOT \\
				\blsamegroup & 176 & 2\,006 & 2\,626 & 1\,752 & 16\,388 & 3\,145 \\
				\blkatzmass & 934 & 16\,589 & 1\,873 & 8\,218 & 129\,337 & 1\,508\,358 \\
				\algfast & 4 & 87 & 190 & 191 & 4\,337 & 737 \\
				\bottomrule
		\end{tabular}}
		\caption{Number of added edges ($|\Delta|$).}
		\label{tab:Katz-edits}
	\end{subtable}

	\begin{subtable}[t]{\linewidth}
		\centering\resizebox{\linewidth}{!}{			\begin{tabular}{lrrrrrr}
				\toprule
				\textbf{Algorithm} & Blogs & Hopkins & Retweet & Deezer & Penn & Pokec \\
				\midrule
				\blgapgreedy & 1.23 & 2942.68 & OOT & OOT & OOT & OOT \\				
				\blsamegroup & 0.03 & 3.43 & 0.47 & 6.68 & 83.78 & 11.60 \\
				\blkatzmass & 0.03 & 0.43 & 0.23 & 0.67 & 3.94 & 95.51 \\
				\algfast & 0.00 & 0.13 & 0.46 & 0.91 & 25.77 & 101.09 \\
				\bottomrule
		\end{tabular}}
		\caption{Running time in seconds.}
		\label{tab:runtime}
	\end{subtable}

	\medskip
	\caption{\algfast edit counts under unconstrained and two-hop admissible
		edge sets. Both reach the target on all datasets.}
	\label{tab:twohop-blade}
	\small
	\resizebox{\columnwidth}{!}{		\begin{tabular}{lrrrrrrr}
			\toprule
			\textbf{Admissible set}
			& Blogs & Hopkins & Retweet & Deezer & Penn   & Pokec & Total \\
			\midrule
			All missing
			& 4 & 87 & 190 & 191 & 4\,337 & 737 & 5\,546 \\
			2-hop
			& 4 & 87 & 198 & 214 & 4\,337 & 761 & 5\,601 \\
			\midrule
			\textbf{Difference}
			& 0 & 0 & 8 & 23 & 0 & 24 & 55 \\
			\bottomrule
		\end{tabular}	}
\end{table}

\noindent
\underline{\textbf{RQ2 Efficiency:}}
\Cref{tab:runtime,tab:pr_runtime} report the running times of the Katz edge-addition methods and the PageRank-based baselines. \algfast reaches exact parity with the fewest edits on all six datasets and is fastest on Blogs and Hopkins. Although some baselines are faster, \blsamegroup uses more edits and misses parity on Penn, while \blkatzmass often fails to improve the top-$k$ objective.
\algfast stays within the time limit on all six datasets, taking $25.77$ seconds on Penn and $101.09$ seconds on Pokec, while reaching the target everywhere. Pokec is the slowest case for \algfast because each score update is performed on by far the largest graph, with over $1.6$ million nodes and $30$ million edges.
Among the Katz edge-addition methods, \algfast provides the best trade-off
between runtime, success, and edit cost. The PageRank baselines are fast on
the smaller graphs but do not scale uniformly in our implementations:
\blfairgd takes $1{,}893.20$ seconds on Pokec, while both LFPR variants
terminate due to insufficient memory. Their runtimes are not directly
comparable to \algfast because they do not solve the Katz edge-addition
problem.

\textbf{Effect of batching.} To measure the runtime benefit of batching, we
evaluate \algslow, an unbatched variant, in the scalability study. \algslow
follows the same boundary-link rule as \algfast but adds a single edge per
iteration and updates Katz scores after every edit. It produces the same edit counts as \algfast on the evaluated datasets, but removes the batching mechanism.
Compared with \algfast, the difference is small on the smaller graphs, where update costs are negligible, but becomes substantial on the larger networks. On Penn and Pokec, \algslow is about 5.1$\times$ and 5.9$\times$ slower, taking 131.17 and 597.80 seconds, respectively. This confirms that \algfast’s batching is not needed for edit quality, but is important for scalability.

\begin{table}[t]
	\centering
\caption{Results for the PageRank baselines for $k=100$. OOM denotes an out-of-memory failure.}
	\label{tab:k100-table-1-PR-reduction}
	
	\begin{subtable}[t]{\linewidth}
		\centering\resizebox{\linewidth}{!}{			\begin{tabular}{lrrrrrr}
				\toprule
				\textbf{Algorithm} & Blogs & Hopkins & Retweet & Deezer & Penn & Pokec \\
				\midrule
				\blfairgd & -525.0 & 0.0 & -300.0 & 0.0 & 0.0 & 0.0 \\
				\blfairwalk & -300.0 & 0.0 & -300.0 & 23.4 & 85.9 & -1244.4 \\
				\blcrosswalk & -800.0 & 55.6 & -125.0 & 85.9 & 85.9 & -800.0 \\
				\bllfprn & 100.0 & 100.0 & -3500.0 & 90.2 & 100.0 & OOM \\
				\bllfpru & 75.0 & 88.9 & -5525.0 & 80.9 & 93.8 & OOM \\
				\bottomrule
			\end{tabular}		}
		\caption{Reduction of top-$k$ unfairness wrt.~the PR ranking (\%).}
		\label{tab:pr_unfairness}
	\end{subtable}
	
	\begin{subtable}[t]{\linewidth}
		\centering\resizebox{\linewidth}{!}{			\begin{tabular}{lrrrrrr@{}}
				\toprule
				\textbf{Algorithm} & Blogs & Hopkins & Retweet & Deezer & Penn & Pokec \\
				\midrule
				\blfairgd & 0.17 & 1.76 & 0.75 & 1.32 & 52.93 & 1893.20 \\
				\blfairwalk & 0.00 & 0.00 & 0.00 & 0.00 & 0.02 & 0.83 \\
				\blcrosswalk & 0.01 & 0.02 & 0.01 & 0.03 & 0.24 & 224.86 \\
				\bllfprn & 0.02 & 0.11 & 9.29 & 10.04 & 6.00 & OOM \\
				\bllfpru & 0.04 & 0.87 & 9.67 & 21.25 & 63.67 & OOM \\
				\bottomrule
			\end{tabular}		}
		\caption{Running time in seconds.}
		\label{tab:pr_runtime}
	\end{subtable}
\end{table}

\noindent
\underline{\textbf{RQ3 Sensitivity to graph structure and parameters:}}
We evaluate robustness, proxy accuracy, Katz parameter $\alpha$, frontier width $q$, top-$k$ cutoff, and target proportion $\pi$.

\textbf{Ranking robustness.}
For $k=100$, \Cref{tab:ranking_preservation} shows the top-100 structure is stable: although membership changes are expected because \algfast promotes nodes across the top-$k$ boundary, the mean Overlap@100 is $0.85$. Among nodes shared by the original and post-intervention top-100 sets, Kendall $\tau$ is at least $0.943$ and averages $0.975$. \algfast also preserves the full converged Katz ranking, with Spearman correlation above $0.9999$ on every dataset.

\begin{figure*}[!t]
\centering
\begin{minipage}[t]{0.48\textwidth}
\vspace{0pt}
\centering
\begin{minipage}[t]{0.48\linewidth}
	\centering
	\hrule height 0pt
	\captionsetup{hypcap=false}
	\captionof{table}{Ranking robustness. Overlap@100 is the fraction of nodes
		shared by original and post-intervention top-100. Kendall $\tau$ is
		computed on their common nodes.}
	\label{tab:ranking_preservation}
	\small
	\setlength{\tabcolsep}{4pt}
	\resizebox{\linewidth}{!}{		\begin{tabular}{lcc}
			\toprule
			\textbf{Dataset} & Overlap@100 & Kendall $\tau$ \\
			\midrule
			Blogs   & 0.98 & 0.997 \\
			Retweet & 0.84 & 0.962 \\
			Deezer  & 0.76 & 0.961 \\
			Hopkins & 0.94 & 0.992 \\
			Penn    & 0.70 & 0.943 \\
			Pokec   & 0.88 & 0.998 \\
			\midrule
			\textbf{Mean} & 0.85 & 0.975 \\
			\bottomrule
		\end{tabular}}
\end{minipage}\hfill\begin{minipage}[t]{0.48\linewidth}
	\centering
	\hrule height 0pt
	\captionsetup{hypcap=false}
	\includegraphics[height=0.72\linewidth]{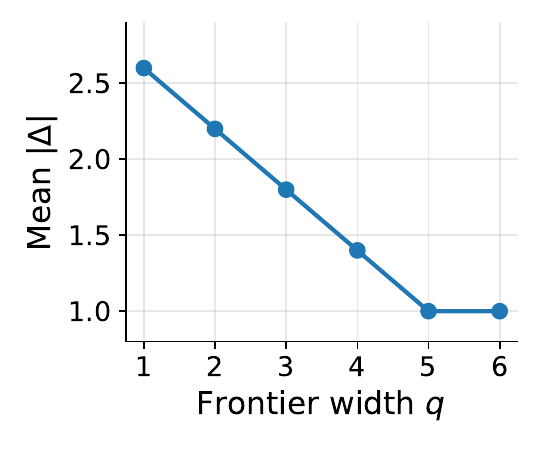}
	\captionof{figure}{Effect of \algfast's boundary width $q$ on the mean number of added edges.}
	\Description{A line chart shows mean edit count decreasing as boundary width increases from one to six.}
	\label{fig:frontierqtable}
\end{minipage}

\medskip
\captionsetup{hypcap=false}
\captionof{table}{Relative proxy error (\%) on Blogs as the Katz parameter $\alpha$ varies. We report the mean over all evaluated candidates and the error of the selected candidate.}
\label{tab:proxy-error}
\setlength{\tabcolsep}{4pt}
\resizebox{\linewidth}{!}{	\begin{tabular}{lrrrrrrrrr}
		\toprule
		& \multicolumn{9}{c}{$c$ for $\alpha=c/\rho(A)$}\\
		\cmidrule(lr){2-10}
		& $0.1$ & $0.2$ & $0.3$ & $0.4$ & $0.5$ & $0.6$ & $0.7$ & $0.8$ & $0.9$ \\
		\midrule
		\textbf{Mean}     & 0.004 & 0.043 & 0.201 & 0.217 & 0.322 & 0.559 & 0.722 & 1.207 & 2.628 \\
		\textbf{Selected} & 0.004 & 0.047 & 0.197 & 0.224 & 0.335 & 0.565 & 0.725 & 1.259 & 2.728 \\
		\bottomrule
	\end{tabular}}
\end{minipage}\hfill\begin{minipage}[t]{0.48\textwidth}
	\vspace{0pt}
	\centering
	\refstepcounter{figure}
	\begin{subfigure}{0.48\linewidth}
		\centering
		\includegraphics[width=\linewidth]{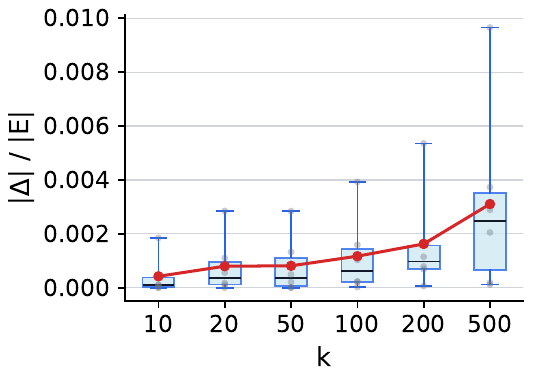}
		\caption{Top-$k$ cutoff.}
		\label{fig:blade-k-variant1}
	\end{subfigure}
	\hfill	\begin{subfigure}{0.48\linewidth}
		\centering
		\includegraphics[width=\linewidth]{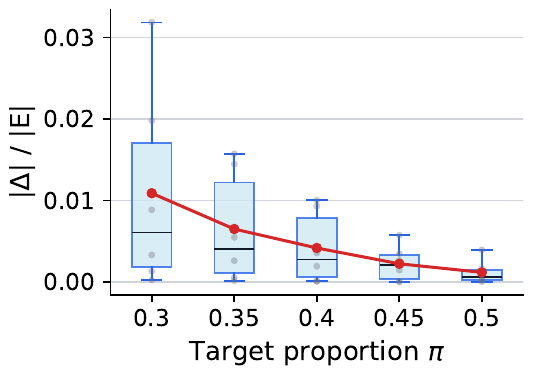}
		\caption{Target distribution $\pi$.}
		\label{fig:blade-target-variant1}
	\end{subfigure}
	\addtocounter{figure}{-1}
	\captionsetup{hypcap=false}
	\captionof{figure}{\algfast sensitivity: Mean ratio of added edges wrt.~size of $E$ over all real-world datasets.}
	\Description{Two line charts show the mean edit-to-edge ratio as the top-k cutoff and target group proportion vary.}
	\label{fig:blade-k-target-sensitivity}

	\medskip
	\captionsetup{hypcap=false}
	\captionof{table}{Effect of Katz parameter $\alpha$, measured by the number of edges added by
		\algfast to reach fairness.}
	\label{tab:alpha}
	\resizebox{\linewidth}{!}{		\begin{tabular}{lrrrrrrrrr}
			\toprule
			& \multicolumn{9}{c}{$c$ for $\alpha=c/\rho(A)$}\\
			\cmidrule(lr){2-10}
			\textbf{Dataset} & $0.1$ & $0.2$ & $0.3$ & $0.4$ & $0.5$ & $0.6$ & $0.7$ & $0.8$ & $0.9$ \\
			\midrule
			Blogs   & 4   & 6   & 7   & 16  & 25  & 41  & 58  & 86  & 128 \\
			Hopkins & 72  & 90  & 70  & 74  & 78  & 116 & 158 & 190 & 310 \\
			Retweet & 131 & 98  & 80  & 65  & 67  & 79  & 94  & 119 & 154 \\
			Deezer  & 237 & 191 & 159 & 136 & 109 & 97  & 92  & 81  & 100 \\
			Penn & 3\,734 & 2\,929 & 2\,381 & 1\,923 & 1\,627 & 1\,410 & 1\,350 & 1\,402 & 1\,273 \\
			Pokec   & 301 & 113 & 46  & 35  & 21  & 19  & 15  & 26  & 18 \\
			\bottomrule
	\end{tabular}}
\end{minipage}
\end{figure*}

\medskip
\textbf{Proxy accuracy.}
Table~\ref{tab:proxy-error} reports the relative proxy error
across the $\alpha$ range on Blogs.
For increasing $\alpha$,
it remains below $0.05\%$ through $0.2/\rho$.
Near the spectral boundary the error reaches $2.73\%$, but, importantly, the proxy
identifies the same top candidate as the full gain at every round,
so the approximation affects estimated magnitudes but not source selection.

\textbf{Effect of top-$k$ cutoff.}
\Cref{fig:blade-k-variant1} reports the $k$-sensitivity results on real-world
datasets for $\pi=0.5$, with
$k\in\{10,20,50,100,200,500\}$.
The required edit ratio remains low across all tested cutoffs, with even the largest setting staying around 1\%.
The mild increase at larger cutoffs is consistent with the fact that larger $k$ requires more underrepresented-group nodes to cross the top-$k$ boundary to satisfy the target distribution.

\textbf{Effect of target proportion $\pi$.}
To test sensitivity to the fairness target, we vary the protected-group target proportion $\pi$ over $\{0.30,0.35,\ldots,0.50\}$ while fixing the other parameters.
The edit cost is driven primarily by the distance between the requested target and the initial top-$k$ group composition, rather than by $\pi$ monotonically.
For example, Hopkins starts with 44 protected nodes in the top-100: reaching $\pi=0.5$ requires 87 edits, whereas reaching the farther target $\pi=0.3$ requires 483 edits.
Similarly, Blogs starts near parity, so $\pi=0.5$ requires only 4 edits, while $\pi=0.3$ requires 606 edits.
Thus, the larger costs at smaller $\pi$ in \Cref{fig:blade-k-target-sensitivity} reflect larger composition shifts, while the edit ratio remains modest overall.

\textbf{Effect of Katz parameter $\alpha$.}
We also test how the intervention cost changes as $\alpha$ approaches the
Katz stability boundary. For each dataset, we set
$\alpha = c/\rho(A)$ with $c \in \{0.1,\ldots,0.9\}$ and report the number
of added edges needed by \algfast. The results show that the effect of $\alpha$
is dataset-dependent. On Blogs and Hopkins, larger $\alpha$ generally makes
the task harder, requiring more edits as long walks receive more weight. In
contrast, Retweet, Deezer, and Pokec exhibit a U-shaped or decreasing trend:
moderate values of $\alpha$ make boundary nodes easier to promote, while very
small or near-critical values require more edits. Overall, \algfast remains
effective across the full valid range, showing that the method is not tied to
the conservative default choice $\alpha=1/d_{\max}$.

\textbf{Effect of frontier width.}
The width $q$ controls how many near-boundary candidates \algfast evaluates.
As \Cref{fig:blade-q-counterexample} illustrates, larger $q$ can avoid a
myopic choice when the highest-scoring candidate is not the cheapest to
promote. On the real-world benchmarks, edit counts did not change beyond
$q=2$. On $100$ planted random instances with $2000$--$3000$ nodes and
$5500$--$8500$ edges where the top candidate was deliberately not cheapest,
the mean edit count fell from $2.60$ at $q=1$ to $2.20$ at $q=2$ and $1.00$
for $q\ge5$ (\Cref{fig:frontierqtable}). Thus, wider search helps in
adversarial configurations but quickly saturates, motivating $q=2$ by default.

\section{Conclusion}

We introduced the Fair Top-$k$ Katz Centrality Design Problem, where the goal
is to achieve target group representation in the Katz top-$k$ set using as
few admissible edge additions as possible. We proved strong hardness results,
derived Katz sensitivity formulas, and developed \algfast for scalable targeted
graph interventions. Experiments show that \algfast improves top-$k$
representation with few edits (relative to graph size) while preserving ranking stability.
Future work includes extending the framework to other centrality measures, multi-group targets, and richer graph interventions.

\section*{Ethical Considerations}

Graph edits can alter exposure and opportunity for individuals even when they
improve aggregate group representation. The formulation also assumes that
group labels, target proportions, and admissible interventions are legitimately
available, which may raise privacy, consent, and governance concerns in a
deployment. The method should therefore not be the sole basis for high-stakes
decisions. Practitioners should define targets with affected stakeholders,
restrict edits to consented and application-valid actions, audit individual and
subgroup outcomes, and retain human oversight.

\appendix

\section{Omitted Proofs}
\label{app:proofs}
\begin{proof}[Proof of \Cref{thm:dual-inapprox-katz}]
	Fix a constant $\psi\ge 1$.
	We give a gap-preserving reduction from \textsc{Independent Set} on $3$-regular graphs.
	Let $(H=(U,F),\ell)$ be an instance, where the question is whether $H$ contains an independent set of size at least $\ell$.
	Write $n_H:=|U|$.
	
	\smallskip
	\noindent\textbf{Choice of the gap parameter.}
	Choose an even integer $Q$ such that
	\[
	Q > \psi \ell
	\qquad\text{and}\qquad
	Q\ge 8 .
	\]
	Since $\psi$ is a fixed constant, $Q$ is polynomial in the input size.
	
	\smallskip
	\noindent\textbf{Katz scores.}
	We use the Katz centrality vector
	\[
	s_G = (I-\alpha A(G)^\top)^{-1}\mathbf{1},
	\]
	with $\alpha=\tfrac12$.
	Since the graph constructed below is a DAG of depth at most $3$, all Katz scores can be computed by summing walks of length at most $3$.
	
	\smallskip
	\noindent\textbf{Construction.}
	We construct a directed, unweighted graph $G$ with two groups $\red$ and $\blue$.
	
	For each vertex $u_i\in U$, create:
	\begin{itemize}
		\item a selector tail $t_i$, colored $\blue$;
		\item a candidate node $r_i$, colored $\red$.
	\end{itemize}
	We make each selector tail have score exactly $Q+1$ by adding $2Q$ private leaf nodes
	\[
	y_{i,1},\dots,y_{i,2Q},
	\]
	colored $\blue$, with fixed edges
	\[
	y_{i,a}\to t_i
	\qquad
	\text{for all } a\in[2Q].
	\]
	Thus
	\[
	s(t_i)=1+\alpha(2Q)=Q+1.
	\]
	
	For each edge $e=\{u_i,u_j\}\in F$, create $\ell$ conflict clones
	\[
	q_{e,1},\dots,q_{e,\ell},
	\]
	colored $\blue$, and add fixed edges
	\[
	r_i\to q_{e,h},
	\qquad
	r_j\to q_{e,h},
	\qquad
	\text{for every } h\in[\ell].
	\]
	
	Next create $\ell$ blue buffer nodes
	\[
	w_1,\dots,w_\ell .
	\]
	Each buffer node $w_a$ is constructed to have Katz score
	\[
	s(w_a)=\frac Q2+\frac34 .
	\]
	To achieve this, add $Q-2$ private leaves
	\[
	d_{a,1},\dots,d_{a,Q-2}
	\]
	with fixed edges $d_{a,b}\to w_a$, and add one intermediate node $p_a$ with one private leaf $z_a\to p_a$ and one fixed edge $p_a\to w_a$.
	Then $s(p_a)=1+\alpha=\tfrac32$, and therefore
	\[
	s(w_a)
	=
	1+\alpha(Q-2)+\alpha\cdot \frac32
	=
	1+\frac{Q-2}{2}+\frac34
	=
	\frac Q2+\frac34 .
	\]
	
	Finally, create $\ell$ fallback red nodes
	\[
	f_1,\dots,f_\ell,
	\]
	colored $\red$.
	For each fallback node $f_a$, create $Q$ private blue leaves
	\[
	g_{a,1},\dots,g_{a,Q}.
	\]
	The edges $g_{a,b}\to f_a$ will be admissible edits rather than fixed edges.
	
	\smallskip
	\noindent\textbf{Admissible edits.}
	The admissible action set $\mathcal A$ consists of two types of edges:
	\begin{itemize}
		\item core selector edges
		\[
		t_i\to r_i
		\qquad
		\text{for every } u_i\in U;
		\]
		\item fallback edges
		\[
		g_{a,b}\to f_a
		\qquad
		\text{for every } a\in[\ell],\ b\in[Q].
		\]
	\end{itemize}
	There are no other admissible edits.
	
	\smallskip
	\noindent\textbf{Top-$k$ parameters.}
	Set
	\[
	k:=n_H+\ell,
	\qquad
	\pi:=\frac{\ell}{k},
	\qquad
	\varepsilon:=\frac{1}{4k^2}.
	\]
	Because the number of red nodes in the top-$k$ is an integer, any nonzero deviation from the target red count $\ell$ is at least $1/k$ in proportion.
	Hence
	\[
	\Unfair(T_k(G\oplus\Delta))<\varepsilon
	\]
	implies
	\[
	|T_k(G\oplus\Delta)\cap V_{\red}|=\ell .
	\]
	Thus feasibility for the dual instance is equivalent to achieving exactly $\ell$ red nodes in the top-$k$ set.
	
	\smallskip
	\noindent\textbf{Score levels.}
	Let $G'$ be any graph obtained by adding admissible edits.
	For a core candidate $r_i$, there is at most one admissible incoming selector edge $t_i\to r_i$.
	Thus:
	\[
	s(r_i)
	=
	\begin{cases}
		1, & \text{if } t_i\to r_i \text{ is not added},\\[2mm]
		1+\alpha s(t_i)
		=
		1+\frac{Q+1}{2}
		=
		\frac Q2+\frac32, & \text{if } t_i\to r_i \text{ is added}.
	\end{cases}
	\]
	We say that $u_i$ is selected if the edge $t_i\to r_i$ is added.
	
	Consider a conflict clone $q_{e,h}$ for $e=\{u_i,u_j\}$.
	If neither endpoint is selected, then
	\[
	s(q_{e,h})=1+\alpha(1+1)=2.
	\]
	If exactly one endpoint is selected, then
	\[
	s(q_{e,h})
	=
	1+\alpha\left(\frac Q2+\frac32+1\right)
	=
	\frac Q4+\frac94 .
	\]
	If both endpoints are selected, then
	\[
	s(q_{e,h})
	=
	1+\alpha\left(
	\frac Q2+\frac32
	+
	\frac Q2+\frac32
	\right)
	=
	\frac Q2+\frac52 .
	\]
	
	For a fallback node $f_a$, if exactly $r$ of its $Q$ private fallback edges have been added, then
	\[
	s(f_a)=1+\frac r2 .
	\]
	In particular, if $r<Q$, then
	\[
	s(f_a)\le 1+\frac{Q-1}{2}=\frac Q2+\frac12,
	\]
	whereas if all $Q$ private fallback edges are added, then
	\[
	s(f_a)=1+\frac Q2=\frac Q2+1.
	\]
	
	For $Q\ge 8$, the relevant score levels satisfy the strict ordering
	\begin{align}
		Q+1
		&>
		\frac Q2+\frac52
		>
		\frac Q2+\frac32
		>
		\frac Q2+1\notag\\
		&>
		\frac Q2+\frac34
		>
		\frac Q2+\frac12
		>
		\frac Q4+\frac94
		>
		2
		>
		1.\label{eq:gap-score-ordering}
	\end{align}
	Therefore the $n_H$ selector tails $t_i$ are always the highest-scoring nodes in the graph and hence always belong to the top-$k$ set.
	The remaining $\ell$ positions of the top-$k$ set are determined by the nodes below the selector tails.
	
	\smallskip
	\noindent\textbf{YES case.}
	Suppose $H$ contains an independent set $S\subseteq U$ with $|S|\ge \ell$.
	Choose any subset $S'\subseteq S$ of size exactly $\ell$.
	Add the $\ell$ selector edges
	\[
	\{\,t_i\to r_i : u_i\in S'\,\}.
	\]
	Every selected red candidate $r_i$ has score $\frac Q2+\frac32$.
	Since $S'$ is independent, every conflict clone has at most one selected endpoint and therefore has score at most
	\[
	\frac Q4+\frac94
	<
	\frac Q2+\frac34
	<
	\frac Q2+\frac32 .
	\]
	The buffer nodes have score $\frac Q2+\frac34$, and all fallback nodes are unactivated and have score $1$.
	Thus, after the $n_H$ selector tails, the next $\ell$ highest-scoring nodes are precisely the $\ell$ selected red candidates.
	Consequently,
	\[
	|T_k(G\oplus\Delta)\cap V_{\red}|=\ell,
	\]
	and hence
	\[
	\Unfair(T_k(G\oplus\Delta))=0<\varepsilon .
	\]
	Therefore
	\[
	\OPT_\varepsilon\le \ell .
	\]
	
	\smallskip
	\noindent\textbf{NO case.}
	Suppose $H$ contains no independent set of size $\ell$.
	We show that no feasible edit set of size smaller than $Q$ exists.
	
	Let $\Delta$ be any edit set with $|\Delta|<Q$.
	Since each fallback node requires all $Q$ of its private fallback edges to reach score $\frac Q2+1$, no fallback node is activated by $\Delta$.
	Indeed, every fallback node has score at most
	\[
	\frac Q2+\frac12
	<
	\frac Q2+\frac34,
	\]
	which is below the blue buffer score.
	
	Let
	\[
	S_\Delta
	:=
	\{\,u_i\in U : t_i\to r_i\in\Delta\,\}
	\]
	be the set of core vertices selected by $\Delta$.
	
	If $|S_\Delta|<\ell$, then fewer than $\ell$ red core candidates have score above the buffer nodes, and no fallback red node has score above the buffers.
	Thus, among the $\ell$ top-$k$ positions below the selector tails, fewer than $\ell$ are red.
	Hence the top-$k$ set does not contain exactly $\ell$ red nodes.
	
	Now suppose $|S_\Delta|\ge \ell$.
	Since $H$ has no independent set of size $\ell$, the set $S_\Delta$ cannot be independent.
	Therefore there exists an edge $e=\{u_i,u_j\}\in F$ with both endpoints in $S_\Delta$.
	For this edge $e$, all $\ell$ conflict clones
	\[
	q_{e,1},\dots,q_{e,\ell}
	\]
	have score
	\[
	\frac Q2+\frac52,
	\]
	which is strictly larger than the score $\frac Q2+\frac32$ of every selected red candidate.
	After the $n_H$ selector tails, the next $\ell$ highest-scoring nodes are these $\ell$ blue conflict clones.
	Thus the top-$k$ set again does not contain exactly $\ell$ red nodes.
	
	In both cases, $\Delta$ is infeasible.
	Since $\Delta$ was arbitrary with $|\Delta|<Q$, every feasible solution in the NO case has size at least $Q$:
	\[
	\OPT_\varepsilon\ge Q .
	\]
	
	\smallskip
	\noindent\textbf{Feasibility of the NO instances.}
	The constructed dual instance is always feasible.
	Indeed, by adding all $Q$ private fallback edges into each of the $\ell$ fallback nodes, each fallback node obtains score
	\[
	\frac Q2+1
	>
	\frac Q2+\frac34,
	\]
	which is above the buffer score.
	If no core selector edges are added, then all conflict clones have score $2$, and the $\ell$ activated fallback nodes occupy the $\ell$ positions below the $n_H$ selector tails.
	Therefore the resulting top-$k$ set contains exactly $\ell$ red nodes.
	Hence $\OPT_\varepsilon<\infty$ for every constructed instance.
	
	\smallskip
	\noindent\textbf{Gap and approximation contradiction.}
	We have shown:
	\[
	\text{YES instance:}\qquad
	\OPT_\varepsilon\le \ell,
	\]
	whereas
	\[
	\text{NO instance:}\qquad
	\OPT_\varepsilon\ge Q.
	\]
	By construction, $Q>\psi\ell$.
	
	Now suppose there were a polynomial-time $\psi$-approximation algorithm for the dual problem.
	On a YES instance, it would return a feasible edit set of size at most
	\[
	\psi\OPT_\varepsilon
	\le
	\psi\ell
	<
	Q.
	\]
	On a NO instance, every feasible edit set has size at least $Q$, so the algorithm must return a solution of size at least $Q$.
	Thus, by checking whether the returned solution has size smaller than $Q$, we could decide whether $H$ contains an independent set of size at least $\ell$.
	
	This would solve \textsc{Independent Set} on $3$-regular graphs in polynomial time, contradicting $\mathrm{P}\ne\mathrm{NP}$.
	Therefore no polynomial-time $\psi$-approximation exists unless $\mathrm{P}=\mathrm{NP}$.
\end{proof}

\begin{proof}[Proof of \Cref{cor:decision-nphard-katz}]
	We assume rational input parameters with polynomial bit complexity; Katz scores can then be computed exactly by solving a rational linear system, and score comparisons can be performed in polynomial time.
	The problem is in NP: given an edit set $\Delta$, we can check in
	polynomial time whether $|\Delta|\le C$ and whether the edited graph
	satisfies the target fairness constraint.
	NP-hardness follows from the gap construction in
	Theorem~3.2.
	In particular, distinguishing
	whether $\OPT_\varepsilon\le \ell$ or $\OPT_\varepsilon\ge Q$ is
	NP-hard, and therefore so is the threshold decision problem.
\end{proof}

\begin{proof}[Proof of \Cref{prop:monotone}]
	We have
	$K_G = \sum_{k=0}^{\infty} \alpha^k A_G^k$ and $K_{G'} = \sum_{k=0}^{\infty} \alpha^k A_{G'}^k$.
	Adding the edge $u \to v$ yields $A_{G'} = A_G + E$,
	where $E$ is the matrix with a single $1$ in position $(u,v)$ and zeros elsewhere. Hence $A_{G'} \ge A_G$ entrywise.
	For nonneg\-ative matrices, entrywise inequality implies $\rho(A_{G'}) \ge \rho(A_G)$,
	so $\alpha < 1/\rho(A_{G'}) \le 1/\rho(A_G)$ and both Neumann series converge.
	We prove $A_{G'}^k \ge A_G^k$ entrywise for all $k \ge 0$ by induction.
	For $k = 0$, both are the identity matrix. Assuming $A_{G'}^k \ge A_G^k$ entrywise,
	\[
	A_{G'}^{k+1} = A_{G'} \cdot A_{G'}^k
	\ge A_G \cdot A_{G'}^k
	\ge A_G \cdot A_G^k = A_G^{k+1},
	\]
	where the first inequality uses $A_{G'} \ge A_G$ and the second uses the induction hypothesis.
	Multiplying by $\alpha^k \ge 0$ and summing over $k$ yields $K_{G'} \ge K_G$ entrywise.
\end{proof}

\begin{proof}[Proof of \Cref{lem:katz-sensitivity}]
	Adding $u \to v$ gives $A_{G'} = A_G + \mathbf{e}_u \mathbf{e}_v^\top$.
	By Sherman--Morrison~\cite{sherman1950adjustment} with $M = I - \alpha A_G$,
	$\mathbf{x} = \alpha \mathbf{e}_u$, $\mathbf{y} = \mathbf{e}_v$:
	\begin{align*}
		K_{G'} = (M - \mathbf{x}\mathbf{y}^\top)^{-1}
		&= K_G + \frac{\alpha \, K_G \mathbf{e}_u \, \mathbf{e}_v^\top K_G}
		{1 - \alpha \, \mathbf{e}_v^\top K_G \mathbf{e}_u}
		\\&= K_G + \frac{\alpha}{1 - \alpha K_G(v,u)} \, \mathbf{k}_{\cdot u}\, \mathbf{k}_{v \cdot}.
	\end{align*}
	The denominator is positive: since $\alpha < 1/\rho(A_{G'})$, the rank-one matrix $\mathbf{k}_{\cdot u}\mathbf{k}_{v\cdot}$ is nonneg\-ative by Proposition~4.1, so $\alpha/(1-\alpha K_G(v,u))>0$, giving $1-\alpha K_G(v,u)>0$.
	
	Finally, we have
	\begin{align*}
		\Delta s(a \mid u \to v)
		&= \mathbf{1}^\top (K_{G'} - K_G) \mathbf{e}_a \\
		&= \frac{\alpha (\mathbf{1}^\top \mathbf{k}_{\cdot u})
			(\mathbf{k}_{v \cdot} \mathbf{e}_a)}{1 - \alpha K_G(v,u)} \\
		&= \frac{\alpha \, s_G(u) \, K_G(v,a)}{1 - \alpha K_G(v,u)}.
	\end{align*}
\end{proof}

\begin{proof}[Proof of \Cref{prop:factorization}]
	Apply Equation~(3) to each term in Definition~4.3 and factor.
	$\lambda_G(u,v)>0$ since $\alpha>0$, $s_G(u)>0$, and $1-\alpha K_G(v,u)>0$ under the Katz-validity assumption.
\end{proof}

\begin{proof}[Proof of \Cref{prop:maxincrease}]
	The two identities follow directly from the single-edge Katz update formula
	with head node \(a\). Katz validity implies nonnegativity of the kernel,
	\(K_G(a,a)\ge 1\), \(K_G(a,u)\ge 0\), and
	\(1-\alpha K_G(a,u)>0\). Hence
	$
	\frac{K_G(a,a)}{1-\alpha K_G(a,u)}\ge 1,
	$
	so the proxy does not overestimate the direct gain of \(a\). The sign of the
	differential gain follows from the second identity because
	\(\widehat g_G(u,a)>0\) and the denominator is positive.
\end{proof}

\section{Population-Proportional Targets}

Katz methods use the Katz ranking, whereas the PageRank-oriented methods use
their resulting PageRank ranking; consequently, their initial top-$k$ counts
can differ. PageRank-oriented methods do not add
edges, so their edit counts are marked ``--.'' OOT denotes the one-hour limit,
and OOM denotes termination due to insufficient memory. Penn \textsc{GapGreedy}
was stopped after $2{,}438$ seconds, and Pokec \textsc{GapGreedy} was not started
after three consecutive one-hour timeouts; both are marked OOT.

\begin{table*}[t]
  \centering
  \caption{Population-proportional target results on Blogs.}
  \label{tab:population-targets-blogs}
  \scriptsize
  \resizebox{\textwidth}{!}{  \begin{tabular}{rlrrrrrrr}
    \toprule
    \textbf{Target} & \textbf{Algorithm} &
    \textbf{Edits} & \textbf{Initial $|T_k\cap V_r|$} &
    \textbf{Final $|T_k\cap V_r|$} & \textbf{Success} &
    $\boldsymbol{\Unfair(T_k)}$ & \textbf{Overlap@100} & \textbf{Time (s)} \\
    \midrule
    52 & \textsc{Blade} & 32 & 48 & 52 & Yes & 0.0000 & 0.96 & 0.004 \\
    52 & \textsc{SameGroup} & 1596 & 48 & 52 & Yes & 0.0000 & 0.93 & 0.224 \\
    52 & \textsc{KatzMass} & 820 & 48 & 49 & No & 0.0009 & 0.99 & 0.011 \\
    52 & \textsc{GapGreedy} & 36 & 48 & 52 & Yes & 0.0000 & 0.96 & 12.245 \\
    52 & \textsc{FairGD} & -- & 48 & 50 & No & 0.0004 & 0.98 & 0.170 \\
    52 & \textsc{FairWalk} & -- & 48 & 47 & No & 0.0025 & 0.86 & 0.000 \\
    52 & \textsc{CrossWalk} & -- & 48 & 47 & No & 0.0025 & 0.68 & 0.009 \\
    52 & \textsc{LFPR-N} & -- & 48 & 56 & No & 0.0016 & 0.80 & 0.018 \\
    52 & \textsc{LFPR-U} & -- & 48 & 54 & No & 0.0004 & 0.90 & 0.037 \\
    \bottomrule
  \end{tabular}}
\end{table*}

\begin{table*}[t]
  \centering
  \caption{Population-proportional target results on Hopkins.}
  \label{tab:population-targets-hopkins}
  \scriptsize
  \resizebox{\textwidth}{!}{  \begin{tabular}{rlrrrrrrr}
    \toprule
    \textbf{Target} & \textbf{Algorithm} &
    \textbf{Edits} & \textbf{Initial $|T_k\cap V_r|$} &
    \textbf{Final $|T_k\cap V_r|$} & \textbf{Success} &
    $\boldsymbol{\Unfair(T_k)}$ & \textbf{Overlap@100} & \textbf{Time (s)} \\
    \midrule
    55 & \textsc{Blade} & 228 & 44 & 55 & Yes & 0.0000 & 0.89 & 0.339 \\
    55 & \textsc{SameGroup} & 2998 & 44 & 55 & Yes & 0.0000 & 0.89 & 4.887 \\
    55 & \textsc{KatzMass} & 83402 & 44 & 55 & Yes & 0.0000 & 0.89 & 1.062 \\
    55 & \textsc{GapGreedy} & OOT & 44 & OOT & OOT & OOT & OOT & OOT \\
    55 & \textsc{FairGD} & -- & 53 & 56 & No & 0.0001 & 0.97 & 1.745 \\
    55 & \textsc{FairWalk} & -- & 53 & 58 & No & 0.0009 & 0.93 & 0.002 \\
    55 & \textsc{CrossWalk} & -- & 53 & 60 & No & 0.0025 & 0.70 & 0.022 \\
    55 & \textsc{LFPR-N} & -- & 53 & 60 & No & 0.0025 & 0.87 & 0.110 \\
    55 & \textsc{LFPR-U} & -- & 53 & 55 & Yes & 0.0000 & 0.89 & 0.948 \\
    \bottomrule
  \end{tabular}}
\end{table*}

\begin{table*}[t]
  \centering
  \caption{Population-proportional target results on Retweet.}
  \label{tab:population-targets-retweet}
  \scriptsize
  \resizebox{\textwidth}{!}{  \begin{tabular}{rlrrrrrrr}
    \toprule
    \textbf{Target} & \textbf{Algorithm} &
    \textbf{Edits} & \textbf{Initial $|T_k\cap V_r|$} &
    \textbf{Final $|T_k\cap V_r|$} & \textbf{Success} &
    $\boldsymbol{\Unfair(T_k)}$ & \textbf{Overlap@100} & \textbf{Time (s)} \\
    \midrule
    61 & \textsc{Blade} & 516 & 34 & 61 & Yes & 0.0000 & 0.73 & 1.093 \\
    61 & \textsc{SameGroup} & 4562 & 34 & 61 & Yes & 0.0000 & 0.71 & 0.819 \\
    61 & \textsc{KatzMass} & 100940 & 34 & 20 & No & 0.1681 & 0.85 & 6.045 \\
    61 & \textsc{GapGreedy} & OOT & 34 & OOT & OOT & OOT & OOT & OOT \\
    61 & \textsc{FairGD} & -- & 48 & 48 & No & 0.0169 & 0.99 & 0.743 \\
    61 & \textsc{FairWalk} & -- & 48 & 57 & No & 0.0016 & 0.90 & 0.001 \\
    61 & \textsc{CrossWalk} & -- & 48 & 52 & No & 0.0081 & 0.33 & 0.011 \\
    61 & \textsc{LFPR-N} & -- & 48 & 74 & No & 0.0169 & 0.72 & 8.349 \\
    61 & \textsc{LFPR-U} & -- & 48 & 70 & No & 0.0081 & 0.77 & 9.906 \\
    \bottomrule
  \end{tabular}}
\end{table*}

\begin{table*}[t]
  \centering
  \caption{Population-proportional target results on Deezer.}
  \label{tab:population-targets-deezer}
  \scriptsize
  \resizebox{\textwidth}{!}{  \begin{tabular}{rlrrrrrrr}
    \toprule
    \textbf{Target} & \textbf{Algorithm} &
    \textbf{Edits} & \textbf{Initial $|T_k\cap V_r|$} &
    \textbf{Final $|T_k\cap V_r|$} & \textbf{Success} &
    $\boldsymbol{\Unfair(T_k)}$ & \textbf{Overlap@100} & \textbf{Time (s)} \\
    \midrule
    44 & \textsc{Blade} & 117 & 26 & 44 & Yes & 0.0000 & 0.82 & 0.606 \\
    44 & \textsc{SameGroup} & 1562 & 26 & 44 & Yes & 0.0000 & 0.77 & 5.842 \\
    44 & \textsc{KatzMass} & 30318 & 26 & 23 & No & 0.0441 & 0.97 & 1.197 \\
    44 & \textsc{GapGreedy} & OOT & 26 & OOT & OOT & OOT & OOT & OOT \\
    44 & \textsc{FairGD} & -- & 34 & 34 & No & 0.0100 & 1.00 & 1.342 \\
    44 & \textsc{FairWalk} & -- & 34 & 30 & No & 0.0196 & 0.94 & 0.002 \\
    44 & \textsc{CrossWalk} & -- & 34 & 41 & No & 0.0009 & 0.67 & 0.029 \\
    44 & \textsc{LFPR-N} & -- & 34 & 34 & No & 0.0100 & 0.68 & 10.394 \\
    44 & \textsc{LFPR-U} & -- & 34 & 35 & No & 0.0081 & 0.74 & 25.355 \\
    \bottomrule
  \end{tabular}}
\end{table*}

\begin{table*}[t]
  \centering
  \caption{Population-proportional target results on Penn.}
  \label{tab:population-targets-penn}
  \scriptsize
  \resizebox{\textwidth}{!}{  \begin{tabular}{rlrrrrrrr}
    \toprule
    \textbf{Target} & \textbf{Algorithm} &
    \textbf{Edits} & \textbf{Initial $|T_k\cap V_r|$} &
    \textbf{Final $|T_k\cap V_r|$} & \textbf{Success} &
    $\boldsymbol{\Unfair(T_k)}$ & \textbf{Overlap@100} & \textbf{Time (s)} \\
    \midrule
    48 & \textsc{Blade} & 3666 & 20 & 48 & Yes & 0.0000 & 0.72 & 23.366 \\
    48 & \textsc{SameGroup} & 16388 & 20 & 45 & No & 0.0009 & 0.75 & 82.865 \\
    48 & \textsc{KatzMass} & 1400650 & 20 & 35 & No & 0.0169 & 0.85 & 26.565 \\
    48 & \textsc{GapGreedy} & OOT & 20 & OOT & OOT & OOT & OOT & OOT \\
    48 & \textsc{FairGD} & -- & 58 & 58 & No & 0.0100 & 1.00 & 50.198 \\
    48 & \textsc{FairWalk} & -- & 58 & 53 & No & 0.0025 & 0.93 & 0.023 \\
    48 & \textsc{CrossWalk} & -- & 58 & 47 & No & 0.0001 & 0.68 & 0.237 \\
    48 & \textsc{LFPR-N} & -- & 58 & 45 & No & 0.0009 & 0.72 & 6.236 \\
    48 & \textsc{LFPR-U} & -- & 58 & 48 & Yes & 0.0000 & 0.75 & 70.913 \\
    \bottomrule
  \end{tabular}}
\end{table*}

\begin{table*}[t]
  \centering
  \caption{Population-proportional target results on Pokec.}
  \label{tab:population-targets-pokec}
  \scriptsize
  \resizebox{\textwidth}{!}{  \begin{tabular}{rlrrrrrrr}
    \toprule
    \textbf{Target} & \textbf{Algorithm} &
    \textbf{Edits} & \textbf{Initial $|T_k\cap V_r|$} &
    \textbf{Final $|T_k\cap V_r|$} & \textbf{Success} &
    $\boldsymbol{\Unfair(T_k)}$ & \textbf{Overlap@100} & \textbf{Time (s)} \\
    \midrule
    49 & \textsc{Blade} & 512 & 38 & 49 & Yes & 0.0000 & 0.89 & 84.532 \\
    49 & \textsc{SameGroup} & 2559 & 38 & 49 & Yes & 0.0000 & 0.89 & 10.324 \\
    49 & \textsc{KatzMass} & OOT & 38 & OOT & OOT & OOT & OOT & OOT \\
    49 & \textsc{GapGreedy} & OOT & 38 & OOT & OOT & OOT & OOT & OOT \\
    49 & \textsc{FairGD} & -- & 53 & 53 & No & 0.0016 & 1.00 & 1889.701 \\
    49 & \textsc{FairWalk} & -- & 53 & 60 & No & 0.0121 & 0.87 & 0.844 \\
    49 & \textsc{CrossWalk} & -- & 53 & 58 & No & 0.0081 & 0.85 & 232.533 \\
    49 & \textsc{LFPR-N} & OOM & 53 & OOM & OOM & OOM & OOM & OOM \\
    49 & \textsc{LFPR-U} & OOM & 53 & OOM & OOM & OOM & OOM & OOM \\
    \bottomrule
  \end{tabular}}
\end{table*}

\end{document}